\documentclass[11pt]{article} 

\usepackage{hyperref}
\usepackage{times} 
\usepackage{amsfonts}
\usepackage{amsmath} 
\usepackage{amsthm,bm}
\usepackage{amssymb}
\usepackage{latexsym} 
\usepackage{epsfig}
\usepackage{latexsym} 
\usepackage{verbatim}
\usepackage[mathscr]{euscript}
\usepackage{caption}
\usepackage{color}
\begin{document}
 
\newtheorem{theorem}{Theorem}[section]
\newtheorem{corollary}[theorem]{Corollary}
\newtheorem{proposition}[theorem]{Proposition} 
\newtheorem{problem}[theorem]{Problem}
\newtheorem{lemma}[theorem]{Lemma} 
\newtheorem{remark}[theorem]{Remark}
\newtheorem{observation}[theorem]{Observation}
\newtheorem{defin}[theorem]{Definition} 
\newtheorem{example}[theorem]{Example}
\newtheorem{conj}[theorem]{Conjecture} 
\def\EPR{\hfill $\Box$\linebreak\vskip3mm} 
 \def\EEX{\hfill $\diamond$}
\newcommand{\aanote}[1]{\textcolor{red}{(Amirhosein: #1)}}
\newcommand{\anote}[1]{\textcolor{green}{(Andrei: #1)}}
\newcommand{\TODO}[1]{\textcolor{red}{TODO: #1}}
\newcommand{\DONE}[1]{\textcolor{blue}{#1}}
\newcommand{\vvec}[2]{\begin{pmatrix} #1 \\ \vdots \\ #2 \end{pmatrix}}

\def\Pol{{\sf Pol}} 
\def\mPol{{\sf MPol}} 
\def\Polo{{\sf Pol}_1} 
\def\PPol{{\sf pPol\;}} 
\def\Inv{{\sf Inv}}
\def\mInv{{\sf MInv}} 
\def\Clo{{\sf Clo}\;} 
\def\Con{{\sf Con}} 
\def\concom{{\sf Concom}\;} 
\def\End{{\sf End}\;}
\def\Sub{{\sf Sub}\;} 
\def\Sym{{\sf Sym}} 
\def\Im{{\sf Im}} 
\def\Ker{{\sf Ker}\;} 
\def\H{{\sf H}}
\def\S{{\sf S}} 
\def\D{{\sf P}} 
\def\I{{\sf I}} 
\def\Var{{\sf var}} 
\def\PVar{{\sf pvar}} 
\def\fin#1{{#1}_{\rm fin}}
\def\P{{\sf P}} 
\def\Pfin{{\sf P_{\rm fin}}} 
\def\Id{{\sf Id}}
\def\R{{\rm R}} 
\def\F{{\rm F}} 
\def\Term{{\sf Term}}
\def\var#1{{\sf var}(#1)} 
\def\Sg#1{{\sf Sg}(#1)} 
\def\Sgg#1#2{{\sf Sg}_{#1}(#2)} 
\def\Cg#1{{\sf Cg}(#1)}
\def\Cgg#1#2{{\sf Cg}_{#1}(#2)}
\def\Cen{{\sf Cen}}
\def\tol{{\sf tol}} 
\def\lnk{{\sf lk}} 
\def\rbcomp#1{{\sf rbcomp}(#1)}
\def\Clo#1{\mathsf{Clo}(#1)} 
\def\Clr#1{\mathsf{Clr}(#1)} 
\def\Aut{\mathsf{Aut}}
\def\Stab{\mathsf{Stab}}
\def\const{\mathsf{c}}
\def\Sp{\mathsf{Sp}}
  
\let\cd=\cdot 
\let\eq=\equiv 
\let\op=\oplus 
\let\omn=\ominus
\let\meet=\wedge 
\let\join=\vee 
\let\tm=\times
\def\ldiv{\mathbin{\backslash}} 
\def\rdiv{\mathbin/}
  
\def\typ{{\sf typ}} 
\def\zz{{\un 0}} 
\def\zo{{\un 1}}
\def\one{{\bf1}} 
\def\two{{\bf2}} 
\def\three{{\bf3}}
\def\four{{\bf4}} 
\def\five{{\bf5}}
\def\pq#1{(\zz_{#1},\mu_{#1})}
  
\let\wh=\widehat 
\def\ox{\ov x} 
\def\oy{\ov y} 
\def\oz{\ov z}
\def\of{\ov f} 
\def\oa{\ov a} 
\def\ob{\ov b} 
\def\oc{\ov c}
\def\od{\ov d} 
\def\oob{\ov{\ov b}} 
\def\rx{{\rm x}}
\def\rf{{\rm f}} 
\def\rrm{{\rm m}} 
\let\un=\underline
\let\ov=\overline 
\let\cc=\circ 
\let\rb=\diamond 
\def\ta{{\tilde a}} 
\def\tz{{\tilde z}}
\let\td=\tilde
\let\dg=\dagger
\let\ddg=\ddagger
  
  
\def\zZ{{\mathbb Z}} 
\def\B{{\mathcal B}} 
\def\P{{\mathcal P}}
\def\zL{{\mathbb L}} 
\def\zD{{\mathbb D}}
 \def\zE{{\mathbb E}}
\def\zG{{\mathbb G}} 
\def\zA{{\mathbb A}} 
\def\zB{{\mathbb B}}
\def\zC{{\mathbb C}} 
\def\zF{{\mathbb F}} 
\def\zM{{\mathbb M}} 
\def\zR{{\mathbb R}}
\def\zS{{\mathbb S}} 
\def\zT{{\mathbb T}} 
\def\zN{{\mathbb N}}
\def\zQ{{\mathbb Q}} 
\def\zW{{\mathbb W}} 
\def\bK{{\bf K}}
\def\C{{\bf C}} 
\def\M{{\bf M}} 
\def\E{{\bf E}} 
\def\N{{\bf N}}
\def\O{{\bf O}} 
\def\bN{{\bf N}} 
\def\bX{{\bf X}} 
\def\GF{{\rm GF}} 
\def\cC{{\mathcal C}} 
\def\cA{{\mathcal A}}
\def\cB{{\mathcal B}} 
\def\cD{{\mathcal D}} 
\def\cE{{\mathcal E}} 
\def\cF{{\mathcal F}} 
\def\cG{{\mathcal G}} 
\def\cH{{\mathcal H}}
\def\cI{{\mathcal I}} 
\def\cL{{\mathcal L}} 
\def\cO{{\mathcal O}} 
\def\cP{{\mathcal P}} 
\def\cQ{{\mathcal Q}} 
\def\cR{{\mathcal R}} 
\def\cRY{{\mathcal RY}}
\def\cS{{\mathcal S}} 
\def\cT{{\mathcal T}} 
\def\cU{{\mathcal U}} 
\def\cV{{\mathcal V}} 
\def\cW{{\mathcal W}} 
\def\cZ{{\mathcal Z}} 
\def\oB{{\ov B}}
\def\oC{{\ov C}} 
\def\ooB{{\ov{\ov B}}} 
\def\ozB{{\ov{\zB}}}
\def\ozD{{\ov{\zD}}} 
\def\ozG{{\ov{\zG}}}
\def\tcA{{\widetilde\cA}} 
\def\tcC{{\widetilde\cC}}
\def\tcF{{\widetilde\cF}} 
\def\tcI{{\widetilde\cI}}
\def\tB{{\widetilde B}} 
\def\tC{{\widetilde C}}
\def\tD{{\widetilde D}} 
\def\ttB{{\widetilde{\widetilde B}}}
\def\ttC{{\widetilde{\widetilde C}}}
\def\tba{{\tilde\ba}} 
\def\ttba{{\tilde{\tilde\ba}}}
\def\tbb{{\tilde\bb}} 
\def\ttbb{{\tilde{\tilde\bb}}}
\def\tbc{{\tilde\bc}} 
\def\tbd{{\tilde\bd}}
\def\tbe{{\tilde\be}} 
\def\tbt{{\tilde\bt}}
\def\tbu{{\tilde\bu}} 
\def\tbv{{\tilde\bv}}
\def\tbw{{\tilde\bw}} 
\def\tdl{{\tilde\dl}} 
\def\ocP{{\ov\cP}}
\def\tzA{{\widetilde\zA}} 
\def\tzC{{\widetilde\zC}}
\def\new{{\mbox{\footnotesize new}}}
\def\old{{\mbox{\footnotesize old}}}
\def\prev{{\mbox{\footnotesize prev}}}
\def\oo{{\mbox{\sf\footnotesize o}}}
\def\pp{{\mbox{\sf\footnotesize p}}}
\def\nn{{\mbox{\sf\footnotesize n}}} 
\def\oR{{\ov R}}
  
  
\def\gA{{\mathfrak A}} 
\def\gC{{\mathfrak C}} 
\def\gV{{\mathfrak V}} 
\def\gS{{\mathfrak S}} 
\def\gK{{\mathfrak K}} 
\def\gH{{\mathfrak H}}
\def\gG{{\mathfrak G}}
\def\gP{{\mathfrak P}}
  
\def\ba{{\bf a}} 
\def\bb{{\bf b}} 
\def\bc{{\bf c}} 
\def\bd{{\bf d}} 
\def\be{{\bf e}} 
\def\bbf{{\bf f}} 
\def\bg{{\bf g}}
\def\bh{{\bf h}}
\def\bi{{\bf i}}
\def\bj{{\bf j}}
\def\bm{{\bf m}} 
\def\bo{{\bf o}} 
\def\bp{{\bf p}} 
\def\br{{\bf r}}
\def\bs{{\bf s}} 
\def\bu{{\bf u}} 
\def\bt{{\bf t}} 
\def\bv{{\bf v}} 
\def\bx{{\bf x}}
\def\by{{\bf y}} 
\def\bw{{\bf w}} 
\def\bz{{\bf z}}

\def\bC{{\bf C}}
\def\bR{{\bf R}}
\def\bM{{\bf M}}

\def\ga{{\mathfrak a}} 
\def\oal{{\ov\al}} 
\def\obeta{{\ov\beta}}
\def\ogm{{\ov\gm}} 
\def\oep{{\ov\varepsilon}}
\def\oeta{{\ov\eta}} 
\def\oth{{\ov\th}} 
\def\ovm{{\ov\mu}}
\def\ozero{{\ov0}}
\def\bB{{\bf B}} 
\def\bA{{\bf A}}

  
\def\CCSP{\hbox{\rm c-CSP}} 
\def\CSP{{\rm CSP}} 
\def\Hom{\mathsf{Hom}}
\def\hom{\mathsf{hom}}
\def\EVAL{{\rm EVAL}} 
\def\NCSP{{\rm \#CSP}} 
\def\mCSP{{\rm MCSP}} 
\def\pCSP{{\rm p\hbox{-}CSP}}
\def\FP{{\rm FP}} 
\def\PTIME{{\bf PTIME}} 
\def\Part{\mathsf{Part}} 
\def\GS{\hbox{($*$)}} 
\def\ry{\hbox{\rm r+y}}
\def\rb{\hbox{\rm r+b}} 
\def\Gr#1{{\mathrm{Gr}(#1)}}
\def\Grp#1{{\mathrm{Gr'}(#1)}} 
\def\Grpr#1{{\mathrm{Gr''}(#1)}}
\def\Scc#1{{\mathrm{Scc}(#1)}} 
\def\rel{\EuScript R} 
\def\relo{\EuScript Q}
\def\rela{\EuScript S} 
\def\reli{\EuScript T} 
\def\relp{\EuScript P} 
\def\dep{\mathsf{dep}}
\def\Filt{\mathrm{Ft}}
\def\Filts{\mathrm{Fts}} 
\def\Agr{$\mathbb{A}$}
\def\Al{\mathrm{Alg}}
\def\Sig{\mathrm{Sig}}
\def\strat{\mathsf{strat}}
\def\relmax{\mathsf{relmax}}
\def\srelmax{\mathsf{srelmax}}
\def\Meet{\mathsf{Meet}}
\def\amax{\mathsf{amax}}
\def\umax{\mathsf{umax}}
\def\emin{\mathsf{Z}}
\def\as{\mathsf{as}}
\def\star{\hbox{$(*)$}}
\def\bmal{{\mathbf m}}
\def\Af{\mathsf{Af}}
\let\sqq=\sqsubseteq
\def\maj{\mathsf{maj}}
\def\razm{\mathsf{size}}
\def\Razm{\mathsf{MAX}}
\def\Centr{\mathsf{Center}}
\def\centr{\mathsf{center}}
\def\RG{\mathsf{RG}}
\def\hom{\mathsf{hom}}
\def\inj{\mathsf{inj}}

\def\Horn{{\sf Horn}}
\def\NAND{{\sf NAND}}
\def\orr{{\sf or}}
\def\horn{{\sf horn}}
\def\Compl{{\sf Compl}}
\def\ROR{{\sf OR}}
\def\Cl{{\sf Cl}}
\def\NEQ{{\sf NEQ}}
\def\IMP{{\sf IMP}}
\def\NIMP{{\sf NIMP}}
\def\EQ{{\sf EQ}}
\def\existspr{\exists^{\equiv p}}
\def\ext{\mathsf{\# ext}}
\def\Ext{\mathsf{Ext}}
\def\extp{\mathsf{\#_p ext}}

  
\let\sse=\subseteq 
\def\ang#1{\langle #1 \rangle}
\def\angg#1{\left\langle #1 \right\rangle}
\def\dang#1{\ang{\ang{#1}}} 
\def\vc#1#2{#1 _1\zd #1 _{#2}}
\def\tms{\tm\dots\tm}
\def\zd{,\ldots,} 
\let\bks=\backslash 
\def\red#1{\vrule height7pt depth3pt width.4pt
\lower3pt\hbox{$\scriptstyle #1$}}
\def\fac#1{/\lower2pt\hbox{$\scriptstyle #1$}}
\def\me{\stackrel{\mu}{\eq}} 
\def\nme{\stackrel{\mu}{\not\eq}}
\def\eqc#1{\stackrel{#1}{\eq}} 
\def\cl#1#2{\arraycolsep0pt
\left(\begin{array}{c} #1\\ #2 \end{array}\right)}
\def\cll#1#2#3{\arraycolsep0pt \left(\begin{array}{c} #1\\ #2\\
#3 \end{array}\right)} 
\def\clll#1#2#3#4{\arraycolsep0pt
\left(\begin{array}{c} #1\\ #2\\ #3\\ #4 \end{array}\right)}
\def\cllll#1#2#3#4#5#6{ \left(\begin{array}{c} #1\\ #2\\ #3\\
#4\\ #5\\ #6 \end{array}\right)} 
\def\pr{{\rm pr}}
\def\prp{\mathsf{pr}^p}
\newcommand{\colect}[2]{\{ #1 \}_{{#2}} }
\let\upr=\uparrow 
\def\ua#1{\hskip-1.7mm\uparrow^{#1}}
\def\sua#1{\hskip-0.2mm\scriptsize\uparrow^{#1}} 
\def\lcm{{\rm lcm}} 
\def\perm#1#2#3{\left(\begin{array}{ccc} 1&2&3\\ #1&#2&#3
\end{array}\right)} 
\def\w{$\wedge$} 
\let\ex=\exists
\def\NS{{\sc (No-G-Set)}} 
\def\lev{{\sf lev}}
\let\rle=\sqsubseteq 
\def\ryle{\le_{ry}} 
\def\ryprec{\le_{ry}}
\def\os{\mbox{[}} 
\def\zs{\mbox{]}}
\def\link{{\sf link}}
\def\solv{\stackrel{s}{\sim}} 
\def\mal{\mathbf{m}}
\def\precs{\prec_{as}}
\def\relad{\rel_{\mathsf{alldif}}}
\def\polylog{\mathsf{polylog}}
\def\Gf{\mathsf{Gf}}
\def\Hg{\mathsf{Hg}}
\def\htw{\mathsf{hw}}
\def\fhtw{\mathsf{fhw}}
\def\shtw{\mathsf{shw}}
\def\sth{\mathsf{th}}
\def\Fix{\mathsf{Fix}}
\def\rpi{{\frac{\pi r}q}}
\def\rppi{{\frac{\pi(q-r)}q}}

  
\def\lb{$\linebreak$}  
  
\def\ar{\hbox{ar}} 
\def\Im{{\sf Im}\;} 
\def\deg{{\sf deg}}
\def\id{{\rm id}}
  
\let\al=\alpha 
\let\gm=\gamma 
\let\dl=\delta 
\let\ve=\varepsilon
\let\ld=\lambda 
\let\om=\omega 
\let\vf=\varphi 
\let\vr=\varrho
\let\th=\theta 
\let\sg=\sigma 
\let\Gm=\Gamma 
\let\Dl=\Delta
\let\kp=\kappa
  
  
\font\tengoth=eufm10 scaled 1200 
\font\sixgoth=eufm6
\def\goth{\fam12} 
\textfont12=\tengoth 
\scriptfont12=\sixgoth
\scriptscriptfont12=\sixgoth 
\font\tenbur=msbm10
\font\eightbur=msbm8 
\def\bur{\fam13} 
\textfont11=\tenbur
\scriptfont11=\eightbur 
\scriptscriptfont11=\eightbur
\font\twelvebur=msbm10 scaled 1200 
\textfont13=\twelvebur
\scriptfont13=\tenbur 
\scriptscriptfont13=\eightbur
\mathchardef\nat="0B4E 
\mathchardef\eps="0D3F

\author{Andrei A. Bulatov\thanks{The first author is supported by an NSERC Discovery grant.} and Amirhossein Kazeminia}{ }{}{}{}
\title{Modular Counting over 3-Element and Conservative Domains}
\maketitle
\begin{abstract}
In the Constraint Satisfaction Problem (CSP for short) the goal is to decide the existence of a homomorphism from a given relational structure $\cG$ to a given relational structure $\cH$. If the structure $\cH$ is fixed and $\cG$ is the only input, the problem is denoted $\CSP(\cH)$. In its counting version, $\#\CSP(\cH)$, the task is to find the number of such homomorphisms. The CSP and $\#\CSP$ have been used to model a wide variety of combinatorial problems and have received a tremendous amount of attention from researchers from multiple disciplines.

In this paper we consider the modular version of the counting CSPs, that is, problems of the form $\#_p\CSP(\cH)$ of counting the number of homomorphisms to $\cH$ modulo a fixed prime number $p$. Modular counting has been intensively studied during the last decade, although mainly in the case of graph homomorphisms. Here we continue the program of systematic research of modular counting of homomorphisms to general relational structures. The main results of the paper include a new way of reducing modular counting problems to smaller domains and a study of the complexity of such problems over 3-element domains and over conservative domains, that is, relational structures that allow to express (in a certain exact way) every possible unary predicate. 
\end{abstract}

\newpage

\section{Introduction}

\subparagraph*{The Constraint Satisfaction Problem.}
In a Constraint Satisfaction Problem (CSP) the goal is to decide the existence of a homomorphism from a given relational structure $\cG$ to a given relational structure $\cH$, \cite{Feder98:computational}. If the target structure is fixed and only $\cG$ is the input, the problem called a \emph{nonuniform CSP} and is denoted by $\CSP(\cH)$, \cite{Feder98:computational}. CSPs and nonuniform CSPs in particular have been used to model a wide variety of combinatorial problems across many disciplines. The complexity of problems of the form $\CSP(\cH)$ for finite relational structures has been thoroughly investigated \cite{Jeavons:algebraic,Bulatov05:classifying,Barto14:local,Bulatov06:simple,Idziak10:tractability,Bulatov17:dichotomy,Zhuk20:dichotomy}.

In the counting version of the CSP, denoted by $\#\CSP$ (or $\#\CSP(\cH)$ if we are dealing with a nonunifrom problem), the goal is to find the number of homomorphisms from a given structure $\cG$ to a given structure $\cH$, or in the case of a nonuniform problem to a fixed structure $\cH$. Counting CSPs have received much attention starting from the seminal work of Valiant \cite{Valiant79:complexity,Valiant79:computing}, to the classification of Boolean \cite{Creignou96:complexity} counting CSPs and those over graphs \cite{Dyer00:complexity}, to the discovery of the algebraic method for such problems \cite{ref:BULATOV_TowardDichotomy}, to a complexity classification of unweighted counting CSPs \cite{DBLP:journals/jacm/Bulatov13,effective-Dyer-doi:10.1137/100811258}, to a sequence of results on weighted counting CSPs \cite{Dyer09:weighted,Bulatov12:weighted,Goldberg10:complexity,Cai10:graph,Cai14:complexity,Cai16:nonnegative} culminating at a complete characterization of counting CSPs with complex weights by Cai and Chen \cite{cai-complex}. In all the research mentioned above the goal is `exact' counting, that is, finding the exact number of (weighted) homomorphisms between relational structures. Approximate counting is another vast area of research that we do not touch in this paper. Statistical Physics often concerns with the problem of computing partition functions that is closely related to counting CSPs. In this paper we consider another variation of counting problems, counting modulo an integer.

The complexity class the Counting $\CSP$ belongs to is $\#P$, the class of problems of counting accepting paths of polynomial time nondeterministic Turing machines. There are several ways to define reductions between counting problems, but the most widely used ones are parsimonious reductions and Turing (polynomial time) reductions. A \emph{parsimonious reduction} from a counting problem $\#A$ to a counting problem $\#B$ is an algorithm that, given an instance $I$ of $\#A$ produces (in polynomial time) an instance $I'$ of $\#B$ such that the answers to $I$ and $I'$ are equal. A \emph{Turing} (or \emph{polynomial time}) \emph{reduction} is a polynomial time algorithm that solves $\#A$ using $\#B$ as an oracle. 

\subparagraph*{Modular counting.}
While the complexity of and algorithms for exact counting are now well understood, counting modulo an integer has received much attention in the past decade. For an integer $k$ and a relational structure $\cH$ the problem of counting homomorphisms from a given relational structure to $\cH$ modulo $k$ is denoted by $\#_k\CSP(\cH)$. Here (and in almost all papers on the subject) the focus is on prime modulae $p$. 

One of the most important features of relational structures affecting the complexity of the corresponding counting CSP is the automorphism group, $\Aut(\cH)$, of $\cH$. Indeed, it can be easily shown that if $\pi\in\Aut(\cH)$ has order $p$ (or is a \emph{$p$-automorphism}), then $\#_p\CSP(\cH)$ is equivalent, i.e., parsimoniously interreducuble, with $\#_p\CSP(\cH^\pi)$, where $\cH^\pi$ denotes the substructure of $\cH$ induced by the set $\Fix(\pi)$ of the fixed points of $\pi$. Clearly, such a reduction by $p$-automorphisms can be repeated until the resulting structure is \emph{$p$-rigid}, that is, has no $p$-automorphisms. Such a derivative $p$-rigid structure is unique up to an isomorphism, as it is shown in \cite{ref:CountingMod2Ini}, and will be denoted by $\cH^{*p}$.

A related problem that will be important in this paper is computing partition functions. Let $M$ be a $k\tm k$-matrix with entries from (in the most general case) a semiring $\zS$. Let $\EVAL(M)$ denote the problem of computing the function 
\[
Z_M(G)=\sum_{\vf:V\to [k]}\prod_{(u,v)\in E}M[\vf(u),\vf(v)]
\]
for a given (di)graph $G=(V,E)$, called the \emph{partition function} of $G$. The complexity of partition functions has been extensively studied in Statistical Physics and Theoretical Computer Science, see \cite{Welsh90:computational,Barvinok16:combinatorics,Bulatov05:partition} for some samples. In the case of real or complex matrices the complexity of exact computing of partition functions is well understood \cite{Bulatov05:partition,Cai10:graph}. The case of matrices over a finite field that arises from modular counting is wide open and presents challenges that can be glimpsed from the `cancellation phenomenon' observed in \cite{Goldberg10:complexity,cai-complex} and produced by elements of the field of a finite multiplicative order.

\subparagraph*{Existing results.}
A systematic research on modular counting was initiated by Faben in \cite{faben2008complexity}, who introduced the complexity classes $\#_pP$ of problems of counting, modulo $p$, accepting paths of polynomial time Turing machines. In the same paper he also proved some key hardness results and gave a complexity classification of $\#_p\CSP(\cH)$ for 2-element structures $\cH$ into polynomial time solvable and $\#_pP$-complete ones. The notion of reduction between modular counting problems is similar to that for exact counting: it is either polynomial time reduction, in the same sense as before, or parsimonious reductions, where the answers to the two instances must be equal modulo $p$. 

Since then most of the effort has been directed at problems $\#_p\CSP(\cH)$ where $\cH$ is a graph. Faben and Jerrum \cite{ref:CountingMod2Ini} investigated the case when $p=2$ and $\cH$ is a tree, and, apart from classifying the complexity of the problem in this case, they also posed a conjecture stating that if $\cH$ is a $p$-rigid graph, $\#_p\CSP(\cH)$ has the same complexity as the exact counting problem $\#\CSP(\cH)$. In a series of papers \cite{ref:CountingModPToTrees_gbel_et_al_LIPIcs,ref:CountingMod2ToSquarefree,ref:CountingMod2ToCactus,ref:CountingModPToK33free,ref:CountingMod2ToKMinor,ref:CountingModPToSquarefree} the Faben/Jerrum conjecture was confirmed for a number of graph classes and values of $p$, and then was proved in full generality by Bulatov and Kazeminia in \cite{DBLP:conf/stoc/BulatovK22}.

While modular counting of graph homomorphisms provides some insight into the more general case, many of the complications do not occur in that case as it was demonstrated in \cite{Kazeminia25:modular}. As we will be using some of the results from this paper, we discuss it in some detail. The first observation is that the more general version of the Faben/Jerrum conjecture is not true for general relational structures: for any prime $p$ there exists a $p$-rigid structure $\cH$ such that $\#\CSP(\cH)$ is hard, but $\#_p\CSP(\cH)$ is easy. This can be fixed to some extent by introducing and exploiting multi-sorted relational structures with a richer structure of automorphisms. We will be using this framework as well, see Section~\ref{sec:structures}. Although it is possible to modify the Faben/Jerrum conjecture for general multi-sorted structures in such a way that no counterexample is known so far, there is still little understanding of the general case.

There are two properties of graphs that made the result of \cite{DBLP:conf/stoc/BulatovK22} possible but fail for general relational structures. One is the structure of automorphisms of direct products of graphs. Due to the results of \cite{ref:ProductOfGraphs,DBLP:conf/stoc/BulatovK22} every automorphism of a direct product of graphs essentially boils down to a combination of automorphisms of the factors and a permutation of the factors. This is no longer true in the general case, even in the case of digraphs, as was demonstrated in \cite{Kazeminia25:modular}. This leads to a failure of the second important property: reductions through primitive positive (pp-) definitions. Let $\cH$ be a relational structure. Recall that a relation $\rel(\vc xk)$ on $H$, the base set of $\cH$, is said to be \emph{pp-definable} in $\cH$ if there is a first order formula that represents $\rel$ and has the following form
\[
\rel(\vc xk)=\exists\vc ym\ \Phi(\vc xk,\vc ym),
\]
where $\Phi$ is a conjunction of predicates from $\cH$ and equality predicates. If $\cH$ is a $p$-rigid graph, then the problem $\#_p\CSP(\cH+\rel)$, where $\cH+\rel$ denotes the \emph{extension} of $\cH$ by adding the predicate $\rel$, is polynomial time reducible to $\#_p\CSP(\cH)$. Such reductions are known to be a powerful tool in the CSP research, as they allow for the use of polymorphisms, that is, homomorphisms from powers $\cH^n$ to $\cH$, to describe the complexity of constraint problems, see \cite{Bulatov05:classifying,ref:POlymorphismAndUsethem_barto2017polymorphisms}. In particular, the existence of a so-called \emph{Mal'tsev polymorphism}, that is, a polymorphism $m:\cH^3\to\cH$ satisfying the equations $m(x,y,y)=m(y,y,x)=x$, makes it possible to represent solutions to a CSP in a compact form that then in some cases can be used to find their number \cite{Bulatov06:simple,effective-Dyer-doi:10.1137/100811258}. In the case of general relational structures pp-definitions no longer give rise to reductions between problems and need to be replaced with their modular counterpart. A relation $\rel(\vc xk)$ is said to be \emph{$p$-mpp-definable} in $\cH$ for a prime $p$, if can be represented by a formula 
\[
\rel(\vc xk)=\existspr\vc ym\ \Phi(\vc xk,\vc ym),
\]
where the requirements on $\Phi$ are the same as before, and which is true for certain values of $\vc xk$ if and only if the number of values of $\vc ym$ that make $\Phi(\vc xk,\vc ym)$ true is not divisible by $p$. Then \cite{Kazeminia25:modular} proves that $\#_p\CSP(\cH+\rel)$ is polynomial time reducible to $\#_p\CSP(\cH)$ for any relational structure $\cH$ and any relation $\rel$ $p$-mpp-definable in $\cH$. 

Unfortunately, $p$-mpp-definitions cannot be captured by polymorphisms, and overcoming this difficulty will be an ongoing theme of this paper. In some cases, however, polymorphisms of $p$-mpp-definable relations play an important role. For instance, in \cite{Kazeminia25:modular} it was shown that if the set of 2-mpp-definable relations of a structure $\cH$ has a Mal'tsev polymorphism, then $\#_2\CSP(\cH)$ can be solved in polynomial time.

\subparagraph*{Our contribution.}
In this paper we build upon the results of \cite{Kazeminia25:modular} to obtain a complexity classification of $\#_p\CSP(\cH)$ for two classes of structures $\cH$. 
We start with a new version of an automorphism and the corresponding rigidity condition. 

A binary polymorphism $f$ of a relational structure $\cH$ is said to be an \emph{automorphic polynomial} if for any $a\in H$, $f(a,x)$ is a permutation of $H$. If, for a prime $p$, for any $a\in H$, $f(a,x)$ is the identity mapping or a permutation of order $p$, and $f(a,x)$ is a permutation of order $p$ for  at least one $a\in H$, $f$ is said to be a \emph{$p$-automorphic polynomial}. The existence of a $p$-automorphic polynomial allows one to reduce a modular counting CSP to a CSP over a smaller structure. For a complete version of the following result see Proposition~\ref{pro:p-auto-poly} and Corollary~\ref{cor:3-element-auto-poly}.

\begin{proposition}\label{pro:p-auto-poly-intro}
Let a relational structure $\cH$ have a $p$-automorphic polynomial $f$. Then there is $\cH^f$ such that $\cH^f$ has smaller cardinality than $\cH$ and $\#_p\CSP(\cH)$ is polynomial time reducible to $\#_p\CSP(\cH^f)$. Moreover, if $\cH$ satisfies some additional conditions, $\cH^f$ can be chosen such that $\#_p\CSP(\cH)$ and $\#_p\CSP(\cH^f)$ are polynomial time equivalent.
\end{proposition}

In order to prove Proposition~\ref{pro:p-auto-poly-intro}, given an instance $(V,\cC)$ of $\#_p\CSP(\cH)$ we find permutations of the domain that are specific for each variable from $V$, have order $p$, and such that they map every tuple of every constraint to a tuple from the same constraint (so-called \emph{consistent permutations}). These permutations allow us to reduce the domain $\cH$ to a smaller structure, because similar to automorphisms the elements that are shifted by the permutations can be eliminated from the problem. On the other hand, let $a\in H$ be an element such that $f(a,x)$ is a permutation of order $p$. Then, if permutations above do not exist, it implies that the element $a$ cannot be a part of any solution, and therefore can be eliminated. Consistent permutations are found by constructing a derivative CSP, whose domain is the symmetric group on $H$, and the CSP itself is an instance of a Mal'tsev CSP over that group \cite{Bulatov06:simple,Berkholz15:limitations}.

Recall that a unary relation pp-definable in a structure $\cH$ is said to be a \emph{subalgebra} of $\cH$. In the case of modular counting we use \emph{$p$-subalgebras}, that is, unary relations $p$-mpp-definable in $\cH$. The two types of relational structures we consider here are 
\begin{itemize}
    \item 
    3-element structures, which are the next step after a similar result for 2-element structures from~\cite{faben2008complexity}, and
    \item 
    \emph{$p$-conservative} structures, i.e., structures such that every possible unary relation is a $p$-subalgebra. 
\end{itemize}
\noindent
The same types of the CSP have been major milestones in the study of the decision problem, see \cite{Bulatov06:3-element,Bulatov11:conservative,Barto11:conservative,Bulatov16:conservative}. They introduced several key ideas that were then used to obtain a dichotomy theorem for the general CSP. 

In this paper in both cases our complexity classification is incomplete. The remaining difficulty is the complexity of modular partition functions. When $p=2$ every partition function has a 0-1 matrix, and therefore can be treated as the problem of counting homomorphism to a graph or a bipartite graph. The complexity of this problem is known, see, \cite{DBLP:conf/stoc/BulatovK22}. If $p\ge3$ this is no longer possible and the structure of matrices of partition functions becomes more intricate, requiring methods from finite fields. Some partial results have been obtained in \cite{Kazeminia25:thesis}.

Let $\ang{\cH}$ and $\ang{\cH}_p$ denote the set of all relations that are, respectively, pp-definable and $p$-mpp-definable in $\cH$. In the case of exact counting a necessary condition for tractability of $\#\CSP(\cH)$ is the existence of a Mal'tsev polymorphism of $\cH$, mentioned above. No similar property is known to be true for modular counting. Our two main results claim that for 3-element and $p$-conservative structures $\cH$ if $\ang{\cH}_p$ has no Mal'tsev polymorphism then $\#_p\CSP(\cH)$ is $\#_pP$-hard. If a Mal'tsev polymorphism exists then we identify certain cases, in which the problem is polynomial time solvable.


\begin{theorem}\label{the:conservative-hard-intro}
Let $\cH$ be a $p$-conservative structure and $p$ a prime. If $\ang{\cH}_p$ does not have a Mal'tsev polymorphism then $\#_p\CSP(\cH)$ is $\#_pP$-hard. Otherwise, if $p=2$  then $\#_2\CSP(\cH)$ is solvable in polynomial time.
\end{theorem}

In order to state the result for 3-element structures let $\cH^f$ be the structure from Proposition~\ref{pro:p-auto-poly-intro} constructed from $\cH$ using a $p$-automorphic polynomial $f$. Note that if $\cH$ is a 3-element structure, $\cH^f$ has at most 2 elements. 

\begin{theorem}\label{the:dichotomy-3-element-intro}
Let $\cH$ be a 3-element structure and $p$ a prime. If one of the following conditions holds:
\begin{enumerate}
\item[(a)] 
$\cH$ has a $p$-automorphic polynomial $f$ and $\ang{\cH^f}_p$ has a Mal'tsev polymorphism, or
\item[(b)] 
$\cH$ does not have a $p$-automorphic polynomial, $\ang{\cH}_p$ has a Mal'tsev polymorphism, and $p=2$,
\end{enumerate}
then $\#_p\CSP(\cH)$ is solvable in polynomial time. If
\begin{enumerate}
\item[(1)] 
$\cH$ has a $p$-automorphic polynomial $f$ and $\ang{\cH^f}_p$ does not have a Mal'tsev polymorphism, or
\item[(2)] 
$\cH$ does not have a $p$automorphic polynomial, $\ang{\cH}_p$ does not have a Mal'tsev polymorphism, 
\end{enumerate}
then $\#_p\CSP(\cH)$ is $\#_pP$-complete.
\end{theorem}

Since $\cH^f$ has at most 2 elements, the cases when $\cH$ has a $p$-automorphic polynomial can be dealt with using an extension of the results by Faben \cite{faben2008complexity}, see Proposition~\ref{pro:2-element-obstruction}.

Apart from the introduction of $p$-automorphic, the main technical contribution of the paper is the following. The property of relations with a Mal'tsev polymorphism that makes Mal'tsev polymorphisms useful is rectangularity. A relation $\rel$ (at least binary) is said to be \emph{rectangular} if for any partition of the coordinate set of $\rel$ into two nonempty parts for any tuples $\ba,\bb$ from the first part and $\bc,\bd$ from the second part, if $(\ba, \bc),(\bb, \bc),(\bb, \bd) \in \rel$, then $(\ba, \bd)\in \rel$. If a structure $\cH$ has a Mal'tsev polymorphism then any relation $\rel\in\ang{\cH}$ is rectangular, a property that is sometimes referred to as \emph{strong rectangularity}. It is the lack of rectangularity that often makes a counting CSP hard, as it was shown in \cite{ref:BULATOV_TowardDichotomy}. In the case of modular counting a similar result is also true, although highly nontrivial, provided a non-rectangular relation is in $\cH$, see \cite{DBLP:conf/stoc/BulatovK22}. The major difficulty though has been obtain hardness from non-rectangular relations in $\ang{\cH}$. Here we found a way to use $\ang{\cH}_p$ instead, while still using a Mal'tsev polymorphism.

\section{Preliminaries}
Let $[n]$ denote the set $\{1, 2, \ldots, n\}$. Let $H^n$ be the direct product of the set $H$ with itself $n$ times and $H_1\tms H_n$ the direct product of sets $\vc Hn$. We denote the members of $H^n$ and $H_1\tms H_n$ using bold font, $\mathbf{a} \in H^n$, $\ba\in H_1\tms H_n$. The $i$-th entry of $\mathbf{a}$ is denoted by~$\mathbf{a}[i]$. 
For $I = \{ \vc ik \} \subseteq [n]$, we write $\pr_I \ba$ for the tuple $(\ba[i_1], \dots, \ba[i_k])$, and for a relation $\rel\sse H^n$ or $\rel \subseteq H_1 \times \dots \times H_n$, we define $\pr_I \rel = \{ \pr_I \ba \mid \ba \in \rel \}$. The \emph{arity} of $\rel$, denoted by $\ar(\rel)$, is defined to be $n$, the length of the tuples in $\rel$.
For $\ba \in H_1\tms H_s$, $s\in[n]$, let $\Ext_\rel(\ba) =  \{ \bb 
\in H_{s+1} \tms H_n | (\ba , \bb) \in \rel\}$, and by $\ext_\rel(\ba)$ we denote the cardinality of $\Ext_\rel(\ba)$. (Note that the order of elements in $\ba$ and $\bb$ and $\rel$ might differ, hence we slightly abuse the notation here.) For a prime $p$ we denote the cardinality of $\Ext_\rel(\ba)$ mod $p$ by $\extp_\rel(\ba)$. Moreover, $\prp_{I}\rel$ denotes the set 
$
\{\pr_{I} \ba\mid\ba \in \rel \text{ and } \extp_{\rel}(\pr_{I} \ba) \neq 0  \}
$.
Often, we treat relations $\rel\sse H_1\tms H_n$ as predicates $\rel: H_1\tms H_n\to \{0,1\}$.

\subsection{Relational structures}\label{sec:structures}

We begin by introducing \emph{multi-sorted} sets. Let $H =\colect{H_i}{i\in [k]}$ be a collection of sets. We assume that the sets $\vc Hk$ are disjoint. 
%
A mapping $\vf$ between two multi-sorted sets $G= \colect{G_i}{i\in[k]}$ and $H= \colect{H_i}{i\in[k]}$ is defined as a collection of mappings $\vf=\colect{\vf_i}{i\in[k]}$, where $\vf_i:G_i\to H_i$, that is, $\vf_i$ maps elements of the sort $i$ in $G$ to elements of the sort $i$ in $H$. 
A mapping $\vf = \colect{\vf_i}{i\in[k]}$ from $\colect{G_i}{i\in[k]}$ to $\colect{H_i}{i\in [k]}$ is injective (bijective), if for all $i\in [k]$, $\vf_i$ is injective (bijective).

A \emph{multi-sorted relational signature} $\sg$ over a set of types $[k]$ is a collection of \emph{relational symbols}, a symbol $\rel\in\sg$ is assigned a positive integer $\ell_\rel$, the \emph{arity} of the symbol denoted $\ar(\rel)$, and a tuple $(\vc i{\ell_\rel})\in[k]^{\ell_\rel}$, the \emph{type} of the symbol. A \emph{multi-sorted relational structure} $\cH$ with signature $\sg$ is a multi-sorted set $\{H_i\}_{i \in [k]}$ and an \emph{interpretation} $\rel^\cH$ of each $\rel\in\sg$, where $\rel^\cH$ is a relation or a predicate on $H_{i_1} \times ... \times H_{i_{\ell_\rel}}$. 
The multi-sorted structure $\cH$ is finite if $H$ contains finitely many finite domains,  and $\sg$ is finite. All structures in this paper are finite. The set $H$ is said to be the \emph{base set} or the \emph{universe} of $\cH$. For the base set we will use the same letter as for the structure, only in the regular font. 
Multi-sorted structures with the same signature and types are called \emph{similar}. 

Let $\cH = (\{H_i\}_{i \in [k]}; \rel_1, \dots, \rel_\ell)$ be a relational structure. We use $\Sp(\cH)$ to denote its \emph{spectrum}, that is, an infinite sequence $(n_1,n_2,\dots)$ where $n_j$ is the number of domains of cardinality $j$. If the structure is finite, $\Sp(\cH)$ is essentially a finite sequence, and we will truncate it by removing the trailing zeroes. Let $\prec$ denote the lexicographic order on the set of such sequences assuming later entries are more senior.

Let $\cG,\cH$ be similar multi-sorted structures with signature $\sg$. A \emph{homomorphism} $\vf$ from $\cG$ to $\cH$ is a mapping $\vf:G\to H$ such that for any $\rel\in\sg$ with type $(\vc i{\ell_\rel})$, if $\rel^{\cG}(\vc a{\ell_\rel})$ is true for $(\vc a{\ell_\rel})\in G_{i_1} \times ... \times G_{i_{\ell_\rel}}$, then $\rel^{\cH}(\vf_{i_1}(a_1)\zd\vf_{i_{\ell_\rel}}(a_{\ell_\rel}))$ is true as well. 
The set of all homomorphisms from $\cG$ to $\cH$ is denoted $\Hom(\cG,\cH)$. The cardinality of $\Hom(\cG,\cH)$ is denoted by $\hom(\cG,\cH)$. 
A homomorphism $\vf$ is an \emph{isomorphism} if it is bijective and the inverse mapping $\vf^{-1}$ is a homomorphism from $\cH$ to $\cG$. If $\cH$ and $\cG$ are isomorphic, we write $\cH\cong\cG$. A homomorphism of a structure to itself is called an \emph{endomorphism}, and an isomorphism to itself is called an \emph{automorphism}. As is easily seen, automorphisms of a structure $\cH$ form a group denoted $\Aut(\cH)$. Let $\pi$ be an automorphism of $\cH$. By $\Fix(\pi)$ we denote the collection $\{\Fix(\pi_i) \}_{j\in[k]}$ of sets of fixed points of the $\pi_i$'s. For $\relo\sse H_{i_1}\tm\dots\tm H_{i_\ell}$, we use $\pi(\relo)$ to denote the set $\{\pi(\ba)\mid \ba\in\relo\}$, where $\pi(\ba)=(\pi_{i_1}(\ba[1])\zd\pi_{i_\ell}(\ba[\ell]))$.
For a prime number $p$ we say that $\pi$ has order $p$ or is a \emph{$p$-automorphism} if $\pi$ is not the identity in $\Aut(\cH)$ and has order $p$ in $\Aut(\cH)$. In other words, each of the $\pi_i$'s is either the identity mapping or has order $p$, and at least one of the $\pi_i$'s is not the identity mapping. Structure $\cH$ is said to be \emph{$p$-rigid} if it has no $p$-automorphisms. 

\begin{proposition}[\cite{Kazeminia25:modular}]\label{pro:mult-uniqueness-main}
Let $\cH$ be a multi-sorted  structure and $p$ a prime. Then up to an isomorphism there exists a unique $p$-rigid multi-sorted  structure $\cH^{*p}$ such that for any structure $\cG$ similar to $\cH$ it holds that $\hom(\cG,\cH)=\hom(\cG,\cH^{*p})$.
\end{proposition}

The \emph{direct product} of multi-sorted $\sg$-structures $\cH,\cG$, denoted $\cH\tm\cG$ is the multi-sorted $\sg$-structure with the base set $H\tm G=\colect{H_i \times G_i}{i\in[k]}$, the interpretation of $\rel\in\sg$ is given by $\rel^{\cH\tm\cG}((a_1,b_1)\zd(a_k,b_k))=1$ if and only if $\rel^{\cH}(\vc ak)=1$ and $\rel^{\cG}(\vc bk)=1$. By $\cH^\ell$ we will denote the \emph{$\ell$th power} of $\cH$, that is, the direct product of $\ell$ copies of $\cH$. A \emph{substructure} $\cH'$ of $\cH$ \emph{induced} by a collection of sets $\{H'_i\}_{i\in[k]}$, where $H'_i\sse H_i$, is the relational structure given by $(\colect{H'_i}{i\in[k]};\vc{\rel'}m)$, where $\rel'_j=\rel_j\cap(H'_{i_1}\tm\dots\tm H'_{i_\ell})$ and $(\vc i\ell)$ is the type of $\rel_j$.

\subsection{The CSP and Modular Counting}

There are two standard formulations of the Constraint Satisfaction Problem (CSP).
Let\lb $\cH = (\{H_i\}_{i \in [k]}; \rel^\cH_1\zd\rel^\cH_\ell)$ be a multi-sorted relational structure. The problem $\CSP(\cH)$ asks, given a relational structure $\cG = (\{G_i\}_{i \in [k]}; \rel^\cG_1\zd\rel^\cG_\ell)$ similar to $\cH$ whether there exists a homomorphism $\vf$ from $\cG$ to $\cH$.

The other standard definition does not involve relational structures. Let $H = \{ H\}_{i \in [k]}$ be a multi-sorted domain and $\Gamma$ a set of multi-sorted relations over $H$, called a \emph{constraint language}. An instance $\cP = (V, \tau, \cC)$ of $\CSP(\Gamma)$ consists of:
\begin{itemize}
    \item a finite set of variables $V$,
    \item a type function $\tau: V \to [k]$ assigning each variable a sort,
    \item a finite set of constraints $\cC$, where each constraint is a pair $\langle \bs, \rel \rangle$ with $\rel \in \Gamma$ of arity $m$ and $\bs = (v_1, \ldots, v_m)$ a tuple of variables whose sorts match the type of $\rel$.
\end{itemize}
A \emph{solution} is a mapping $\varphi : V \to \bigcup_{i\in[k]}H_i$ such that for all $v \in V$, $\varphi(v) \in H_{\tau(v)}$, and for every constraint $\langle \bs, \rel \rangle$, the tuple $(\varphi(v_1), \ldots, \varphi(v_m))$ belongs to $\rel$.

For a structure $\cH = (\{H_i\}_{i \in [k]}; \rel^\cH_1\zd\rel^\cH_\ell)$, let $\Gm_\cH=\{\rel^\cH_1\zd\rel^\cH_\ell\}$. It is well known, see e.g.\ \cite{Feder98:computational} and \cite{ref:POlymorphismAndUsethem_barto2017polymorphisms}, that the problems $\CSP(\cH)$ and $\CSP(\Gm_\cH)$ can be easily translated into each other. The same holds for the counting versions $\#\CSP(\cH)$ and $\#\CSP(\Gm_\cH)$, as well as the modular counting versions $\#_p\CSP(\cH)$ and $\#_p\CSP(\Gm_\cH)$. The conversion procedure is as follows.  Given an instance $\cG$ of $\CSP(\cH)$ with signature $\sg$, construct an instance $\cP = (V, \tau, \cC)$ of $\CSP(\Gm_\cH)$ by setting $V = \bigcup_{i \in [k]} G_i$, with the type function $\tau$ induced by the sorts of $\cG$, and for every relation $\rel \in \sg$ and tuple $\bs \in \rel^\cG$, include the constraint $\langle \bs, \rel^\cH \rangle$ in $\cC$. Conversely, for a finite constraint language $\Gm=\{\vc\rel\ell\}$ over $H$ we can construct a structure $\cH_\Gm=(H,\vc\rel\ell)$ with an appropriately selected signature. Then any instance of $\CSP(\Gm)$ can be transformed into an instance of $\CSP(\cH_\Gm)$  by reversing this construction. This correspondence naturally extends to the counting and modular counting versions.

\subsection{Expansions of Relational Structures and the CSP}

A (multi-sorted) relational structure $\cH'=(H;\vc{\rel'}{\ell'})$, $H=\colect{H_i}{i\in[k]}$, is an \emph{expansion} of a relational structure $\cH=(H;\vc\rel\ell)$, if $\ell'\ge\ell$ and $\rel'_i=\rel_i$ for $i\in[\ell]$. It will be convenient to denote $\cH'$ by $\cH+\{\rel_{\ell+1}\zd\rel_{\ell'}\}$ or just $\cH+\rel_{\ell'}$ if $\ell'=\ell+1$. Several types of expansions are often used in the context of the CSP.

Let $\cH=(\{\cH_i\}_{i\in[k]}; \vc{\rel}{\ell})$ be a multi-sorted structure with signature $\sg$ and $\cH^=$ its expansion by adding a family of binary relational symbols $=_{H_i}$ (one for each type) interpreted as the equality relation on $H_i$, $i\in[k]$. A \emph{constant relation} over a set $\{H_i\}_{i\in[k]}$ is a unary relation $C_{H_i,a}=\{a\}$, $a\in H_i,i\in[k]$ (such a predicate can only be applied to variables of type $i$). 
For a structure $\cH$ by $\cH^\const$ we denote the expansion of $\cH$ by all the constant relations $C_{H_i,a}$, $a\in H_i,i\in[k]$. 

\begin{proposition}[\cite{Kazeminia25:modular}]\label{pro:ConstantCSP-main}
Let $\cH$ be a multi-sorted relational structure and $p$ prime.\\[1mm]
(1) $\#_p\CSP(\cH^=)$ is polynomial time reducible to $\#_p\CSP(\cH)$;\\[1mm]
(2) Let $\cH$ be $p$-rigid. Then $\#_p\CSP(\cH^\const)$ is polynomial time reducible to $\#_p\CSP(\cH)$. 
\end{proposition}

Observe that if a relational structure contains all the constant relations, in particular, any structure of the form $\cH^\const$, it is obviously rigid, as every automorphism should respect every $C_{H_i,a}$ and map $a$ to $a$, and therefore also $p$-rigid for any $p$.

Using Proposition~\ref{pro:mult-uniqueness-main} we may assume that all the structures we are dealing with are $p$-rigid. Then, by Proposition~\ref{pro:ConstantCSP-main}(2) we may assume they contain all the constants, that is, $\cH=\cH^\const$. We will use this assumption from now on. 

Primitive-positive definitions play a central role in the CSP research. It has been proved in multiple circumstances that expanding a relational structure with pp-definable relations does not change the complexity of the corresponding CSP. This has been proved for the decision CSP in \cite{Jeavons:algebraic,Bulatov05:classifying}, the exact counting CSP \cite{ref:BULATOV_TowardDichotomy}, and in certain cases (where relational structures are expansions of graphs) \cite{DBLP:conf/stoc/BulatovK22} of modular counting CSP. The reader is referred to \cite{ref:POlymorphismAndUsethem_barto2017polymorphisms} for details about primitive positive definitions and their use in the study of the CSP.

Let $\cH$ be a multi-sorted relational structure with the base set $H$. A \emph{primitive positive} (pp-) formula in $\cH$ is a first-order formula 
\[
\exists \vc ys\Phi(\vc xk,\vc ys),
\]
where $\Phi$ is a conjunction of atomic formulas of the form $z_1=_Hz_2$ or $\rel(\vc z\ell)$, $\vc z\ell$\lb$\in\{\vc xk,\vc ys\}$, and $\rel$ is a predicate of $\cH$. Every variable in $\Phi$ is assigned a type $\tau(x_i),\tau(y_j)$ in such a way that for every atomic formula $z_1=_Hz_2$ it holds that $\tau(z_1)=\tau(z_2)$, and for any atomic formula $\rel(\vc z\ell)$ the sequence $(\tau(z_1)\zd\tau(z_\ell))$ matches the type of $\rel$. We say that $\cH$ \emph{pp-defines} a predicate $\relo$ if there exists a pp-formula such that
\[
\relo(\vc xk)=\exists \vc ys\Phi(\vc xk,\vc ys).
\]
The set of all pp-definable relations in the (multi-sorted) relational structure $\cH$ is denoted by $\langle \cH \rangle$.

For $\ba\in\relo$ by $\ext_\Phi(\ba)$ we denote the number of assignments $\bb\in H_{\tau(y_1)}\tms H_{\tau(y_s)}$ to $\vc ys$ such that $\Phi(\ba,\bb)$ is true. We denote the number of such assignments modulo $p$ by $\extp_\Phi(\ba)$.

As was shown in \cite{Kazeminia25:modular}, pp-definitions are not always compatible with modular counting in the sense that will be made precise later. The concept of pp-definitions adapted to modular counting is that of \emph{modular pp-definitions}, or \emph{$p$-mpp-definitions}, where $p$ is the modulus for counting. For a prime number $p$ the \emph{$p$-modular} quantifier $\existspr$ is interpreted as follows
\begin{align}
\rel(\vc xk)=&\existspr\vc ys \Phi(\vc xk,\vc ys) \text{ is true } \nonumber \\ 
&\Leftrightarrow  \ext_\Phi(\vc xk) \not \equiv 0 \pmod{p}. \label{equ:p-mpp} 
\end{align}
A primitive positive formula that uses $p$-modular quantifiers instead of regular ones is called a \emph{$p$-modular primitive positive} (or \emph{$p$-mpp} for short). The relation it defines is said to be \emph{$p$-mpp-definable}, and the definition itself is called a \emph{$p$-mpp-definition}. 
The set of all $p$-mpp-definable relations in $\cH$ is denoted by $\langle \cH \rangle _p$. Note that modular quantifiers are not as robust as regular ones. In particular, $\exists^{\equiv p}x,y$ and $\exists^{\equiv p}x\ \exists^{\equiv p}y$ may result in a different predicate. The same is of course true for more complicated applications of modular quantifiers, so, one needs to be extra careful with $p$-mpp-definitions.

If for the $p$-mpp-definition~\eqref{equ:p-mpp} $\extp_\Phi(\ba)=1$ for all $\ba\in\rel$, the $p$-mpp definition is said to be \emph{strict}. By Proposition~5.6 from \cite{DBLP:conf/stoc/BulatovK22} if $\rel$ has a $p$-mpp-definition, it has a strict $p$-mpp definition.

\begin{proposition}[\cite{Kazeminia25:modular}]\label{pro:GadgetExists}
Let $\cH$ be a be a $\sg$-structure (single- or multi-sorted), and $p$ a prime. Let $\rel$ be a relation that is $p$-mpp-definable in $\cH$,
then $\#_p\CSP(\cH+\rel)$ is polynomial time reducible to $\#_p\CSP(\cH)$.
In particular, if $\rel$ is conjunctive definable in $\cH$ (that is, $s=0$ in~\eqref{equ:p-mpp}), $\#_p\CSP(\cH+\rel)$ is polynomial time reducible to $\#_p\CSP(\cH)$.
\end{proposition}


Let $\cH = (\{H_i\}_{i \in [k]}; \rel_1, \dots, \rel_m)$ be a multi-sorted relational structure. A subset $A \subseteq H_i$ is called a \emph{subalgebra} (of sort $i$) if it is pp-definable in $\cH$ as a unary relation. If $A$ is $p$-mpp-definable in $\cH$, it is said to be \emph{$p$-subalgebra}. In other words, $A$ is a $p$-subalgebra if there exists a relation  $\relo(x)$ $p$-mpp-definable in  $\cH$, such that $a\in A$ if and only if $\relo(a)$.

\subsection{Polymorphisms and Automorphisms of Products}

We will also use polymorphisms of relational structures and constraint languages. Let $\cH=(H,\vc\rel\ell)$, $H=\{H_i\}_{i\in[k]}$ be a multi-sorted relational structure. An $r$-ary \emph{polymorphism} $f$ of $\cH$ is a homomorphism $f:\cH^r\to\cH$. In other words, a collection of mappings $f=\{f_i\}_{i\in[k]}$, $f_i:H_i^r\to H_i$, such that for every $\rel_j$, $j\in[\ell]$, with type $(\vc is)$, for any $\vc\ba r\in\rel_j$ 
\[
(f_{i_1}(\ba_1[1],\ba_2[1],\dots,\ba_r[1])\zd f_{i_s}(\ba_1[s],\ba_2[s],\dots,\ba_r[s]))\in\rel_j.
\]
If this condition for $\rel_j$ holds, $f$ is also said to be a polymorphism of $\rel_j$. For a constraint language $\Gm$ over $H$ the definition of a polymorphism is essentially the same: $f$ is a polymorphism of $\Gm$ if it is a polymorphism of every relation from $\Gm$.

Let $H=\colect{H_i}{i\in[k]}$ be a multi-sorted set. A collection of mappings $\{e^{s,r}_i\}_{i\in[k]}$, $s\in[r]$, is said to be the \emph{$s$'th $r$-ary projection} if $e^{s,r}_i(\vc ar)=a_s$ for all $\vc ar\in H_i$ and $i\in[k]$. Projections are polymorphisms of any relation, structure, or constraint language. 

The set of polymorphisms of a relational structure $\cH$ or a constraint language $\Gm$ is closed under compositions. Let $g$ be an $r$-ary polymorphism of $\cH$ and $f^1\zd f^r$ $q$-ary polymorphisms of $\cH$. The \emph{composition} of $g$ and $f^1\zd f^r$ is the $q$-ary operation $h$ given by
\[
h_i(\vc xq)=g_i(f^1_i(\vc xq)\zd f^r_i(\vc xq)), \qquad \text{for all $i\in[k]$.}
\]
It is well known \cite{Bulatov05:classifying,ref:POlymorphismAndUsethem_barto2017polymorphisms} that the composition of polymorphisms is again a polymorphism.

A particular type of polymorphisms plays an outsized role in counting CSP research. A ternary polymorphism $m$ of a multi-sorted relational structure $\cH=(H,\vc\rel\ell)$, $H=\{H_i\}_{i\in[k]}$, (a constraint language $\Gm$ over $H$) is said to be \emph{Mal'tsev} if for every $i\in[k]$ the mapping $m_i$ satisfies the equations
$
m_i(x,y,y)=m_i(y,y,x)=x.
$
If $\cH$ has a Mal'tsev polymorphism, the set of solutions of any instance $\cG$ of $\CSP(\cH)$ admits a compact representation by a set of solutions whose number is linear in the size of $G$, see, \cite{Bulatov06:simple,effective-Dyer-doi:10.1137/100811258}. Under certain conditions this allows to solve $\#\CSP(\cH)$ in polynomial time. 

One other concept that we will refer to in this paper is the rectangularity. A binary relation $\rel\subseteq H_1 \times H_2$ is called \emph{rectangular} if $(a, c),(a, d),(b, c) \in\rel$ implies $(b, d) \in\rel$ for any $a, b \in H_1, c, d \in H_2$. A relation $\rel\subseteq H_{i_1} \tms H_{i_n}$ for $n \geq 2$ is rectangular if for every $I\subsetneq[n]$, the relation $R$ is rectangular when viewed as a binary relation, a subset of $\pr_I\rel\tm\pr_{[n]-I}\rel$. Note that rectangularity is a structural property, and it has nothing to do directly with the size of the relation or its parts. A relational structure $\cH$ is \emph{strongly rectangular} if every relation $\rel\in \langle \cH \rangle$ of arity at least 2 is rectangular. 
The following lemma provides a connection between strong rectangularity and Mal'tsev polymorphisms and was first observed in \cite{Hagemann73:permutable} although in a different language.

\begin{lemma}[\cite{Hagemann73:permutable}, see also \cite{effective-Dyer-doi:10.1137/100811258}]
A relational structure is strongly rectangular if and only if it has a Mal’tsev polymorphism.
\end{lemma}

We introduce a modular version of this concept, strongly $p$-rectangular structures.
A relational structure $\cH$ is said to be strongly \emph{$p$-rectangular}, if every $\rel \in \ang \cH _p$ is rectangular. It is shown in \cite{Kazeminia25:modular} that a relational structure can be strongly rectangular, but not strongly $p$-rectangular.

Also, it was shown in~\cite{Kazeminia25:modular} that if a  relational structure is strongly 2-rectangular and admits a Mal'tsev polymorphism, then the corresponding modular counting CSP problem modulo 2 is tractable. This result highlights the algorithmic significance of Mal'tsev polymorphisms: their presence enables efficient algorithms, while their absence plays a central role in establishing hardness. 

\begin{theorem}[\cite{Kazeminia25:modular}]\label{the:parity-algorithm} 
Let $\cH$ be a 2-rigid, strongly 2-rectangular multi-sorted relational structure, and $\ang \cH _2$ has a Mal'tsev polymorphism. Then $\#_2\CSP(\cH)$ can be solved in time $O(n^5)$. 
\end{theorem}

Another part polymorphisms will be playing in this paper is their connection to automorphisms of powers of structures. Let $\cH$ be a multi-sorted structure with the base set $H=\colect{H_i}{i\in[k]}$ and $\cH^r$ its direct power. A mapping $f:\cH^r\to\cH^r$ of $\cH^r$ to itself is a collection of mappings $\colect{f_i}{i\in[k]}:H_i^r\to H_i^r$ each of which can also be represented as $f_i=(f^1_i\zd f^r_i)$, where $f_i^j:H_i^r\to H_i$. Rearranging these mapping we obtain collections $f^j=\colect{f^j_i}{i\in[k]}$, $j\in[r]$, where each $f^j$ is a mapping of $H^r$ to $H$. The following lemma is straightforward from the definitions.

\begin{lemma}\label{lem:poly-auto}
In the notation above $f$ is a homomorphism of $\cH^r$ to itself if and only if $f^j$ is a polymorphism of $\cH$ for every $j\in[r]$. 
\end{lemma}

Clearly, $f$ is an automorphism of $\cH^r$ if and only if the $f^j$ are such that $f$ is bijective.

\section{Automorphic Polynomials}\label{sec:automorphic}

In this section we introduce a new automorphism-like construction that allows one to reduce $\#_p\CSP(\cH)$ to a CSP over a smaller structure that we will extensively use in the future. A binary polymorphism $f$ of a (multi-sorted) relational structure $\cH=(\{H_i\}_{i\in[k]};\vc\rel\ell)$ is said to be an \emph{automorphic polynomial} if 
for any domain $H_i$ of $\cH$ and any $a\in H_i$, $f_i(a,x)$ viewed as a mapping from $H_i$ to $H_i$ is a permutation of $H_i$. If, for a prime $p$, for any domain $H_i$ and any $a\in H_i$, $f_i(a,x)$ is the identity mapping or a permutation of order $p$, and $f_i(a,x)$ is a permutation of order $p$ for at least one domain $H_i$ and at least one $a\in H_i$, $f$ is said to be a \emph{$p$-automorphic polynomial}.

\begin{proposition}\label{pro:p-auto-poly}
Suppose that a relational structure $\cH$ has a $p$-automorphic polynomial $f$. Then there is a structure $\cH^f$ such that $\Sp(\cH^f)\prec\Sp(\cH)$ and $\#_p\CSP(\cH)$ is polynomial time reducible to $\#_p\CSP(\cH^f)$.
\end{proposition}

In order to prove Proposition~\ref{pro:p-auto-poly} we start with an auxiliary construction introduced in \cite{Bulatov12:weighted} and used in the context of modular counting in \cite{Kazeminia25:modular}. 

Let $\cH=(\{H_i\}_{i\in[k]};\vc\rel\ell)$ be a multi-sorted relational structure. We assume that every $H_i$ is among $\vc\rel\ell$ as a unary relation. The structure $b(\cH)$ is constructed as follows:
\begin{itemize}
\item The domains are $\vc\rel\ell$, the relations from $\cH$.
\item For $i \le j$, $s \in [\ar(\rel_i)]$, and $t \in [\ar(\rel_j)]$, define the binary relation $\relo^{ij}_{st} \subseteq \rel_i \times \rel_j$ by
\[
\relo^{ij}_{st} = \{ (\ba, \bb) \mid \ba \in \rel_i,\, \bb \in \rel_j,\, \ba[s] = \bb[t] \}.
\]
This relation expresses coordinate-wise agreement between tuples in $\rel_i$ and $\rel_j$, and is added to $b(\cH)$ to track such identifications across the structure.

\end{itemize}

First, the $\CSP$ over $b(\cH)$ is interreducible with that over $\cH$ and $b(\cH)$ shares several important properties with $\cH$.

\begin{proposition}[\cite{Kazeminia25:modular}]\label{pro:binarization-main-body-appendix}
For any prime $p$ and any structure $\cH$\\[2mm]
(1) $\#_p\CSP(\cH)$ is parsimoniously interreducible with  $\#_p\CSP(b(\cH))$;\\
(2) $\cH$ has a Mal'tsev polymorphism if and only if $b(\cH)$ does;\\[1mm]
(3) $\cH$ is $p$-rigid if and only if $b(\cH)$ is;\\[1mm]
(4) If $\cH$ is strongly $p$-rectangular then so is $b(\cH)$. 
\end{proposition}

The reduction from $\#_p\CSP(\cH)$ to $\#_p\CSP(b(\cH))$ works as follows. We use the alternative definition of the CSP here. Let $\cP=(V,\cC)$ be an instance of $\#_p\CSP(\cH)$, with variable set $V$ and constraint set $\cC$. We assume that for every $v\in V$ with domain $H_i\in H$, $\cP$ contains a constraint $\ang{(v),H_i}$. We construct an instance $b(\cP)$ of $\NCSP(b(\cH))$, which has variable set $V'$ and constraint set $\cC'$, as follows.
\begin{itemize}
    \item 
    For each $C\in\cC$, we introduce a variable $v_C\in V'$. Thus, essentially, $V'=\cC$.
    \item For all $C_1,C_2\in\cC$, $C_1=\ang{(v^1_1\zd v^1_m), \rel_i}, C_2=\ang{(v^2_1\zd v^2_n), \rel_j}$, such that $v^1_s=v^2_t$ for some $s\in[m],t\in[n]$ we introduce a constraint $\ang{(v_{C_1},v_{C_2}),\relo^{ij}_{st}}\in\cC'$. 
\end{itemize}
As is easily seen, for any solution $\vf$ of $\cP$, the mapping $\psi:V'\to \bigcup_{j\in[\ell]}\rel_j$ is a solution of $b(\cP)$, where $\psi$ is defined as follows. For every $v\in V'$ that corresponds to a constraint $\ang{\bs,\rel_j}$, set $\psi(v)=\vf(\bs)$. Conversion of solutions of $\CSP(b(\cH)$ into solutions of $\CSP(\cH)$ is more complicated and we do not need it here.

We are now in a position to prove Proposition~\ref{pro:p-auto-poly}. 

\begin{proof}[Proof of Proposition~\ref{pro:p-auto-poly}.]
We will construct another auxiliary problem from an instance $\cP$ of $\#_p\CSP(\cH)$. Let $\cP=(V,\cC)$ be any instance of $\#_p\CSP(\cH)$ and $b(\cP)=(V',\cC')$ the corresponding instance of $b(\cH)$. Let $\rel_v$ be the domain of $v\in V'$ and $\Sym(\rel_v)$ the symmetric group of $\rel_v$. We introduce an instance $s(\cP)$ with domains $\Sym(\rel_v)$, constraints that are derivative from the constraints of $b(\cP)$ and will allow us to conclude that either a $p$-automorphic polynomial can be used to reduce $\cP$ to a problem with smaller domains, in the sense that there exists at least one variable whose domain becomes smaller, or that $\cP$ does not have solutions that are equal to certain elements of the domains and therefore, again, can be reduced to a problem with smaller domains.

More precisely, $s(\cP)=(V'',\cC'')$ is constructed as follows.
\begin{itemize}
    \item 
    $V''=V'$ and the domain of $v\in V''$ is $\Sym(\rel_v)$;
    \item 
    for every constraint $\ang{(v_{C_1},v_{C_2}),\relo^{ij}_{st}}\in\cC'$, introduce the constraint \lb$\ang{(v_{C_1},v_{C_2}),\rela^{ij}_{st}}\in\cC''$, where
    \begin{align*}
    \rela^{ij}_{st} &=\{(\vf_1,\vf_2)\mid \vf_i\in\Sym(\rel_{v_i}), i\in \{ 1,2\}, \text{ and } (\vf_1(\ba_1),\vf_2(\ba_2))\in \relo^{ij}_{st}\\ 
    & \hspace*{2cm}\text{ for any } (\ba_1,\ba_2)\in \relo^{ij}_{st}\}.
    \end{align*}
\end{itemize}

\noindent
Let $\cP=(V,\cC)$ be a an instance of $\CSP(\cH)$, for each $C\in\cC$, $C=\ang{\bs_C,\rel_C}$, and $H_{\tau(v)}$ the domain of the variable $v\in V$. A collection $\{\vf_C\mid C\in \cC\}$ of permutations $\vf_C:\rel_C\to\rel_C$ is said to be a \emph{consistent collection of permutations} if for any $C_1,C_2$ and any $\ba_1\in\rel_1,\ba_2\in\rel_2$ such that $\pr_{\bs_1\cap\bs_2}\ba_1=\pr_{\bs_1\cap\bs_2}\ba_2$ it holds that $\pr_{\bs_1\cap\bs_2}\vf_{C_1}(\ba_1)=\pr_{\bs_1\cap\bs_2}\vf_{C_2}(\ba_2)$. 

\medskip

{\sc Claim 1.} Each solution of $s(\cP)$ provides a collection $\vf=\{\vf_C\}_{C\in\cC}$ of consistent permutations of tuples from solutions of $\cP$.
\smallskip
\renewcommand{\qedsymbol}{$\blacksquare$}
\begin{proof}[Proof of Claim 1]
    
Let $\psi:V'\to\bigcup_{j\in[\ell]}\Sym(\rel_j)$ be a solution of $s(\cP)$. Then set $\vf_C=\psi(v_C)$. The collection $\vf=\{\vf_C\}_{C\in\cC}$ is consistent. Indeed, let $C_1=\ang{\bs_1,\rel_i},C_2=\ang{\bs_2,\rel_j}\in\cC$ be two constraints and $\ba_1\in\rel_i,\ba_2\in\rel_j$ such that $\pr_{\bs_1\cap\bs_2}\ba_1=\pr_{\bs_1\cap\bs_2}\ba_2$. Then for any $s\in[\ar(\rel_i)],t\in[\ar(\rel_j)]$ such that $\bs_1[s]=\bs_2[t]$, it holds that $(\vf_{C_1},\vf_{C_2})\in\rela^{ij}_{st}$. Therefore, as $(\ba_1,\ba_2)\in\relo^{ij}_{st}$, we have $(\vf_{C_1}(\ba_1),\vf_{C_2}(\ba_2))\in\relo^{ij}_{st}$, that is, $\vf_{C_1}(\ba_1)[s]=\vf_{C_2}(\ba_2)[t]$, as required.
\end{proof}

Next, we propose a way to solve $s(\cP)$.

As every $\Sym(\rel_v)$ is a group, we can introduce a Mal'tsev operation $m=\{m_i\}_{i\in[k]}$ on the collection of $\Sym(\rel_i)$, $i\in[k]$, as follows: $m_i(x,y,z)=xy^{-1}z$, where multiplication and inverse are as in the group $\Sym(\rel_i)$.
\medskip

{\sc Claim 2.} For any $i,j\in[k]$ and any $s\in[\ar(\rel_i)], t\in[\ar(\rel_j)]$, the mapping $m$ is a polymorphism of $\rela^{ij}_{st}$.

\smallskip
\renewcommand{\qedsymbol}{$\blacksquare$}
\begin{proof}[Proof of Claim 2]
Recall that $(\ba_1,\ba_2)\in\relo^{ij}_{st}$ if and only if $\ba_1[s]=\ba_2[t]$. Therefore if $(\vf_1,\vf_2)\in\rela^{ij}_{st}$, for any $\ba_1\in\rel_i,\ba_2\in\rel_j$ such that $\ba_1[s]=\ba_2[t]$ it holds that $\vf_1(\ba_1)[s]=\vf_2(\ba_2)[t]$. Thus, for any $(\vf_{1q},\vf_{2q})\in\rela^{ij}_{st}$, $q=1,2,3$, and any $(\ba_1,\ba_2)\in\rel^{ij}_{st}$ we obtain
\begin{align*}
m(\vf_{11},\vf_{12},\vf_{13})(\ba_1)[s] &=\vf_{13}(\vf^{-1}_{12}(\vf_{11}(\ba_1)))[s]=\vf_{23}(\vf^{-1}_{22}(\vf_{21}(\ba_2)))[t]\\
&=m(\vf_{21},\vf_{22},\vf_{23})(\ba_2)[t],    
\end{align*}
implying $(m(\vf_{11},\vf_{12},\vf_{13}),m(\vf_{21},\vf_{22},\vf_{23}))\in\rela^{ij}_{st}$.
\end{proof}

\renewcommand{\qedsymbol}{$\Box$}
By Claim~2 and the results of \cite{Bulatov06:simple} for any instance $\cP$ of $\#_p\CSP(\cH)$ the problem $s(\cP)$ can be solved in polynomial time. Moreover, as $m$ is idempotent, for any $C_0=\ang{\bs_0,\rel_i}\in\cC$ and $\vf_0\in\Sym(\rel_i)$, the problem $s(\cP)+\{v_{C_0}=\vf_0\}$ can be solved in polynomial time, that is, we can verify in polynomial time whether or not there is a consistent collection of mappings $\vf=\{\vf_C\}_{C\in\cC}$ such that $\vf_{C_0}=\vf_0$. With that in mind the following transformation of $\cP$ can be performed in polynomial time.

Let $i\in[k]$ and $a\in H_i$ be such that $f_i(a,x)$ is a permutation of $H_i$ of order $p$. Then a structure $\cH'$ is constructed as follows:
\begin{itemize}
\item 
The domains are the same as those of $\cH$, except $H_i$. The domain $H_i$ is replaced with two domains: $H'_i=H_i-\{a\}$ and $H''_i=H_i-B$, where $B$ is the union of all nontrivial orbits of $f_i(a,x)$.
\item
Every relation $\rel_j$ containing $i$ in its type is replaced with $\rel'_j$ and $\rel''_j$ that are obtained from $\rel_j$ by replacing $H_i$ with $H'_i$ (respectively $H''_i$) restricting to tuples that do not contain $a$ (respectively, do not contain elements from $B$). 
\end{itemize}
Let $C_0=\ang{\bs_0,\rel_j}\in\cC$ be such that for some $v\in\bs_0$ the type of $v$ is $H_i$, say, $\bs_0[s]=v$, and there is $\ba\in\rel_j$ with $\ba[s]=a$. If there are no such $C_0$ and $\ba$, then $H_i$ can be either completely eliminated of it can be reduced by removing $a$ from it. Thus, $\cP$ itself can be considered as an instance of $\#_p\CSP(\cH')$. Otherwise let $\vf_0\in\Sym(\rel_j)$ be given by $\vf_0(\bx)=f(\ba,\bx)$. By Claim~2, it is possible to verify whether or not there exists a solution $\vf$ of $s(\cP)$ such that $\vf_{C_0}=\vf_0$. If such a solution exists, $H_i$ can be  reduced by eliminating all the nontrivial orbits of $f_i(a,x)$, since for any solution $\psi$ of $\cP$ with $\psi(v)$ belonging to such an orbit, the mappings $\vf(\psi),\vf^2(\psi)\zd\vf^{p-1}(\psi)$ are also different solutions of $\cP$ and together they contribute 0 modulo $p$ into the number of solutions of $\cP$. Suppose such a solution does not exists. As is easily seen, if there is a solution $\psi$ of $\cP$ such that $\psi(\bs_0)=\ba$, then the mappings $\vf_C\in\Sym(\rel)$, $C=\ang{\bs,\rel}\in\cC$, given by $\vf_C(\bx)=f(\psi(\bs),\bx)$, form a solution of $s(\cP)$ with $\vf_{C_0}=\vf_0$. Therefore, there is no such solution $\psi$ and the tuple $\ba$ can be removed from $\rel_j$. Repeating this procedure for every $\ba\in\rel_j$ with $\ba[s]=a$ we can either reduce the domain of $v$ by eliminating nontrivial orbits of $f_i(a,x)$, or by eliminating $a$. In both cases after applying this transformation to all constraints of $\cP$ we obtain an instance of $\#_p\CSP(\cH')$.
\end{proof}

Note that in general Proposition~\ref{pro:p-auto-poly} is not very useful, because $\cH^f$ can be any structure with $\Sp(\cH^f)\prec\Sp(\cH)$ for which $\#_p\CSP(\cH^f)$ is hard. What we are planning to use is the construction from the proof of Proposition~\ref{pro:p-auto-poly} and the following corollary.

\begin{corollary}\label{cor:3-element-auto-poly}
Let $\cH$ be a 3-element single-sorted structure such that there is a 2-automorphic polynomial $f$ of $\cH$ and for $a\in H$ the mapping $f(a,x)$ is a permutation of $H$ of order 2. Then, if $\{a\}$ and $H-\{a\}$ are 2-subalgebras of $\cH$, the structure $\cH^f$ from Proposition~\ref{pro:p-auto-poly} can be chosen such that $\#_p\CSP(\cH)$ and $\#_p\CSP(\cH^f)$ are polynomial time interreducible.
\end{corollary}

\begin{proof}
In order to prove Corollary~\ref{cor:3-element-auto-poly} we only need to show that the problem $\#_p\CSP(\cH^f)$ is polynomial time reducible to $\#_p\CSP(\cH)$. Let $\rel_j$ be a relation of $\cH$ of arity $r$. Then 
\[\rel'_j=\rel_j\cap(\{a\}\tm\dots\tm\{a\})\quad \text{and}\quad \rel''_j=\rel_j\cap((H-\{a\})\tm\dots\tm(H-\{a\})). 
\]
Since $\{a\}$ and $H-\{a\}$ are 2-subalgebras of $\cH$, the relation $\rel'_j,\rel''_j$ are conjunctive definable in $\cH$. The result follows from Proposition~\ref{pro:GadgetExists}.
\end{proof}

\section{Tools for Reduction}\label{sec:tools}

This section introduces the main technical tools used in our reductions. First, 
we revisit the framework of automorphism-stable sets and Möbius inversion, originally developed by Bulatov and Kazeminia~\cite{DBLP:conf/stoc/BulatovK22}. We extend their approach to handle more general cases required in our setting.
Then, we define indicator problems and explain how the existence of a Mal'tsev polymorphism relates to the presence of a specific predicate in the set of all (multi-sorted) pp-definable relations.
Then, we apply all the constructions and tools that we developed in this section to two special classes of relational structures: $p$-conservative and three-element structures.

\subsection{Automorphism-stable Sets}\label{sec:auto-stable}
Let $\cH$ be a (multi-sorted) relational structure and $n$ a natural number. We call a subset $A \subseteq \cH^n$ \emph{automorphism-stable} if there exists $\ba \in A$ such that the set
\[
\Stab(\ba, A) = \{ \pi \in \Aut(\cH^n) \mid \pi(\ba) \in A \}
\]
is a subgroup of $\Aut(\cH^n)$.  Note that $\Stab(\ba, A)$ is always nonempty since it contains the identity automorphism.

Given a collection $\vc\ba s$ of tuples of elements from $\cH^n$, we define \lb $\Stab(\vc\ba s)$ to be the set of automorphisms $\pi \in \Aut(\cH^n)$ such that each $\ba_i$ is a fixed point of $\pi$, i.e.,
\[
\Stab(\vc\ba s) = \{ \pi \in \Aut(\cH^n) \mid \pi(\ba_i) = \ba_i \text{ for all } i = 1, \ldots, s \}.
\]

We also define a combination of these notions: For $\vc\ba s \in \cH^n$ and $A \subseteq \cH^n$, the set
\[
\Stab(\vc\ba s, A) = \{ \pi \in \Aut(\cH^n) \mid \pi(\ba_i) = \ba_i \text{ for all } i, \text{ and } \pi(A) \subseteq A \}
\]
consists of automorphisms that fix each element of $\vc\ba s$ and map $A$ into itself.
We start with a generalization of Lemma~3.8 from \cite{DBLP:conf/stoc/BulatovK22}. First, we introduce a specific type of homomorphism. For (multi-sorted) relational structures $\cG$ and $\cH$, and distinguished elements $x_1,\ldots,x_m \in \cG$ and $a_1,\ldots,a_m \in \cH$, we define 
\[
\Hom\bigl((\cG,x_1,\ldots,x_m),(\cH,a_1,\ldots,a_m)\bigr)
\]
to be the set of all homomorphisms $\varphi : \cG \to \cH$ such that $\varphi$ is a homomorphism from $\cG$ to $\cH$ and, in addition, $\varphi(x_i) = a_i$ for all $i \in [m]$.  
Furthermore, for $A \subseteq \cH$, we set
\[
\Hom\bigl((\cG,x),(\cH,A)\bigr) \;=\; \bigcup_{a \in A} \Hom\bigl((\cG,x),(\cH,a)\bigr).
\]
The counting versions of these definitions are defined analogously.

\begin{lemma}[see Lemma~3.8, \cite{DBLP:conf/stoc/BulatovK22}]\label{lem:mobius-point-appendix}
Let $p$ be prime, $\cH$ a multi-sorted relational structure. 
\begin{itemize}
    \item[(1)]
    Let $A\sse\cH^n$ be an automorphism-stable set. If for every $\cG$ and $x\in\cG$, $\hom((\cG,x),(\cH^n,A))\equiv c\pmod p$, where $c$ does not depend on $\cG$ and $\bx$, then the structure $\cH^n$ has a $p$-automorphism $\pi\in\Stab(\ba,A)$ for some $\ba\in A$.
    \item[(2)]
    Let $\vc\ba s\in\cH$. If for every $\cG$ and $\vc xs\in\cG$, 
    \[
    \hom((\cG,\vc xs),(\cH^n,\vc\ba s))\equiv c\pmod p,
    \]
    where $c$ does not depend on $\cG$ and $\vc xs$, then the structure $\cH^n$ has a $p$-automorphism $\pi\in\Stab(\vc\ba k)$.
    \item[(3)]
    Let $\vc\ba s\in\cH$ and $A\sse\cH^n$. If for every $\cG$ and $\vc xs,x\in\cG$, 
    \[
    \hom((\cG,\vc xs,x),(\cH,\vc\ba s,A))\equiv c\pmod p,
    \]
    where $c$ does not depend on $\cG$ and $\vc xs$, then the structure $\cH^n$ has a $p$-automorphism $\pi\in\Stab(\vc\ba s,A)$.
\end{itemize}
\end{lemma}

\begin{proof}
Parts (1) and (2) follow from Lemma~3.8 of~\cite{DBLP:conf/stoc/BulatovK22}, so we only need to prove item (3). For completeness, we briefly recall the definition of the \emph{M\"obius inversion formula} in the form we require. Let $H$ be a finite set, $\Part(H)$ denote the set of equivalence relations on $H$ ordered with respect to inclusion, and let $M, N : \Part(H) \to \mathbb{Z}$ be functions on $\Part(H)$ such that
\[
M(\theta) = \sum_{\eta \geq \theta} N(\eta).
\]
Then
\[
N(\zeta) = \sum_{\theta \in \Part(H)} w(\theta) \cdot M(\theta),
\]
where the function $w : \Part(H) \to \mathbb{Z}$ is defined recursively as follows:
\begin{itemize}
  \item $w(\zeta) = 1$, where $\zeta$ is the top equivalence relation corresponding to the partition with only one class,
  \item for any partition $\theta < \zeta$, 
  \[
  w(\theta) = -\sum_{\gamma > \theta} w(\gamma).
  \]
\end{itemize}

(3) We use the M\"obius inversion formula on $\Part(H)$. For the multisorted set $H=\{H_i\}_{i\in[k]}$, an equivalence relation is a collection $\{\th_i\}_{i\in[k]}$ of equivalence relations on the $H_i$, and the relation $\gamma\le\th$ is defined by $\gm_i\le\th_i$ for $i\in[k]$.  Let $\ba\in A$ and set 
\begin{align*}
M(\th) &=\hom((\cH/_\th,\ba_1/_\th\zd\ba_k/_\th,\ba/_\th),(\cH,\vc\ba k, A)),\\ 
N(\th) &=\inj((\cH/_\th,\ba_1/_\th\zd\ba_k/_\th,\ba/_\th),(\cH,\vc\ba k,A)).
\end{align*}
Then $N(=_H)\ne0$, as it includes the identity mapping, and we have
\begin{align*}
N(=_H) &=\sum_{\th\in\Part(H)}w(\th) M(\th)\\ &=\sum_{\th\in\Part(H)}w(\th)\hom((\cH/_\th,\ba_1/_\th\zd\ba_k/_\th,\ba/_\th),(\cH,\vc\ba k, A))\\
&\equiv c\cdot\sum_{\th\in\Part(H)}w(\th)\equiv0\pmod p.
\end{align*}
Note that $N(=_H)=\Stab(\vc\ba k,A)$ and therefore is a subgroup of $\Aut(\cH)$. As it has order that is a multiple of $p$, it contains a $p$-automorphism.
\end{proof}

\subsection{Indicator Problem and Indicator Obstruction}\label{sec:indicator}

\newcommand{\vecthree}[3]{\begin{pmatrix} #1 \\ #2 \\ #3 \end{pmatrix}}

In this section, we show that in many cases, if a relational structure $\cH$ is not strongly $p$-rectangular, then there is a rectangularity obstruction of a very specific kind. Recall that for a (multi-sorted) relational structure $\cH = (\{H_i\}_{i\in[k]}; \vc\rel \ell)$, the \emph{ternary indicator problem} $\cI_3(\cH)$ is an instance of $\CSP(\cH)$ that characterizes the ternary polymorphisms of $\cH$.

\begin{itemize}
    \item 
    The set of variables $V$ is the set $H^{(3)} = H_1^3 \cup \dots \cup H_k^3$ of all triples over the domains of $\cH$, i.e., for each triple $\ba = (\ba[1], \ba[2], \ba[3]) \in H_i^3$, the set $V$ contains a variable $x_\ba$ with domain $H_i$. 
    \item 
    For every $s$-ary relation $\rel_j$ of type $(\vc is)$ from $\cH$ and for any tuples $\ba_1, \dots, \ba_s \in H^{(3)}$, where each $\ba_q = (\ba_q[1], \ba_q[2], \ba_q[3]) \in H_{i_q}^3$, we include the constraint
    $
    \left\langle (x_{\ba_1}, \dots, x_{\ba_s}), \rel_j \right\rangle
    $
    in the instance $\cI_3(\cH)$ if for every $t \in [3]$, the tuple $(\ba_1[t], \dots, \ba_s[t])$ belongs to $\rel_j$. 
    That is, the projection of the sequence $(\ba_1, \dots, \ba_s)$ onto each coordinate $t \in [3]$ is a tuple in $\rel_j$.
\end{itemize}

As is easily seen, every solution of $\cI_3(\cH)$ defines a ternary multi-sorted operation $\{f_i\}_{i \in [k]}$ with $f_i: H_i^3 \to H_i$.

\begin{example}\label{exa:indicator1}
(a) Let $\cH_1$ be a relational structure with only one domain $H=\{0,1\}$ and one (binary) relation $\rel=\left(\begin{array}{cc} 0&1\\ 1&0\end{array}\right)$, where columns represent the tuples from $\rel$. Then $\cI_3(\cH_1)$ contains 8 variables $x_{000},x_{001},x_{010},x_{011},x_{100},x_{101}, x_{110},x_{111}$, that is, all possible triples of elements of $\cH$, and the constraints are imposed on every pair $(x_{abc},x_{def})$ such that $(a,d),(b,e),(c,f)\in\rel$. In this case this means all pairs of `dual' variables of the form $(x_{abc},x_{\neg a\neg b\neg c})$. Thus, the constraints are 
\[
\ang{(x_{000},x_{111}),\rel}, \ang{(x_{001},x_{110}),\rel}, \ang{(x_{010},x_{101}),\rel}, \ang{(x_{011},x_{100}),\rel}.
\]

(b) For a structure $\cH_2$ let the domain $H$ be the same set $\{0,1\}$, but the only binary relation is given by $\relo=\left(\begin{array}{ccc} 0&0&1\\ 0&1&0\end{array}\right)$. The set of variables of $\cI_3(\cH_2)$ is the same as before, but the constraints are very different. As is easily seen, the constraints in this case are 
\[
\begin{array}{llll}
\ang{(x_{000},x_{000}),\relo} & \ang{(x_{000},x_{001}),\relo} & \ang{(x_{000},x_{010}),\relo} & \ang{(x_{000},x_{011}),\relo} \\
\ang{(x_{000},x_{100}),\relo} & \ang{(x_{000},x_{101}),\relo} & \ang{(x_{000},x_{110}),\relo} & \ang{(x_{000},x_{111}),\relo} \\
\ang{(x_{001},x_{000}),\relo} & \ang{(x_{001},x_{010}),\relo} & \ang{(x_{001},x_{100}),\relo} & \ang{(x_{001},x_{110}),\relo} \\
\ang{(x_{010},x_{000}),\relo} & \ang{(x_{010},x_{001}),\relo} & \ang{(x_{010},x_{100}),\relo} & \ang{(x_{010},x_{101}),\relo} \\
\ang{(x_{100},x_{000}),\relo} & \ang{(x_{100},x_{001}),\relo} & \ang{(x_{100},x_{010}),\relo} & \ang{(x_{100},x_{011}),\relo} \\
\ang{(x_{011},x_{000}),\relo} & \ang{(x_{011},x_{100}),\relo} & \ang{(x_{101},x_{000}),\relo} & \ang{(x_{101},x_{010}),\relo} \\
\ang{(x_{110},x_{000}),\relo} & \ang{(x_{110},x_{001}),\relo} & \ang{(x_{111},x_{000}),\relo} &
\end{array}
\]
\end{example}

\begin{lemma}[\cite{Jeavons99:expressive}]\label{lem:multisorted-indicator}
For a (multi-sorted) relational structure $\cH$, a ternary (multi-sorted) operation on the domains of $\cH$ is a solution of $\cI_3(\cH)$ if and only if it is a ternary polymorphism of $\cH$. 
\end{lemma}

We will call the set of all solutions of $\cI_3(\cH)$ the \emph{indicator predicate} $\Upsilon_3(\cH)$. It is not difficult to see that $\Upsilon_3(\cH)$ is conjunctive-definable in $\cH$. Indeed, let $\cC$ be the set of all constraints as defined above in $\cI_3(\cH)$. Then we define the indicator predicate as the conjunction of all those constraints
\[
\Upsilon_3(\cH)(\bx) = \bigwedge_{\ang{\bs,\rel_j} \in \cC} \rel_j(\bs),
\]
where $\bx$ is the tuple consisting of all variables $x_{(a,b,c)}$ indexed by triples $(a,b,c) \in H^{(3)}$.

We now define certain coordinate sets used to analyze $\Upsilon_3(\cH)$. Let
\[
I_\cH = \{(a, b, b) \mid a, b \in H_i,\ i \in [k]\}, \quad
J_\cH = \{(b, b, a) \mid a, b \in H_i,\ i \in [k],\ a \ne b\}.
\]
Note that the diagonal triples $(a, a, a)$ are included in $I_\cH$. Enumerate these sets as\lb 
$
I_\cH = \{(a_1, b_1, b_1), \dots, (a_N, b_N, b_N)\}$,
$J_\cH = \{(c_1, c_1, d_1), \dots, (c_M, c_M, d_M)\},
$
and define tuples
\[
\ba_\cH = (a_1, \dots, a_N),\ 
\bb_\cH = (b_1, \dots, b_N),\ 
\bc_\cH = (c_1, \dots, c_M),\ 
\bd_\cH = (d_1, \dots, d_M).
\]
Then 
$\ba_\cH, \bb_\cH \in \pr_{I_\cH}\Upsilon_3(\cH)$, and $\bc_\cH, \bd_\cH \in \pr_{J_\cH} \Upsilon_3(\cH)$. Also, 
\[
(\ba_\cH,\bc_\cH),(\bb_\cH,\bc_\cH),(\bb_\cH,\bd_\cH)\in\pr_{I_\cH\cup J_\cH}\Upsilon_3(\cH).
\]
Moreover, $\cH$ admits a Mal'tsev polymorphism if and only if\lb $(\ba_\cH, \bd_\cH) \in \pr_{I_\cH \cup J_\cH} \Upsilon_3(\cH)$. 

Indeed, this condition tests whether the tuple $(\ba_\cH, \bd_\cH)$, which represents applying the Mal'tsev term coordinate-wise to the elements of each triple from $I_\cH\cup J_\cH$, belongs to the projection of the relation $\Upsilon_3(\cH)$. That is, since 
\[
(\ba_\cH, \bc_\cH),\ (\bb_\cH, \bc_\cH),\ (\bb_\cH, \bd_\cH) \in \pr_{I_\cH \cup J_\cH} \Upsilon_3(\cH),
\]
the closure under a Mal'tsev polymorphism would imply that the corresponding output tuple, $(\ba_\cH, \bd_\cH)$, must also be in the relation. So, its membership characterizes whether such a polymorphism exists.

\begin{example}\label{exa:indicator2}
Let us reconsider the structures $\cH_1,\cH_2$ from Example~\ref{exa:indicator1} and their indicator problems. The following tables contain all the solutions of $\cI_3(\cH_1)$ and $\cI_3(\cH_2)$, that is, the tuples of $\Upsilon_3(\cH_1),\Upsilon_3(\cH_2)$. The rows are labeled with the variables of $\cI_3(\cH_1),\cI_3(\cH_2)$.
\begin{align*}
& \Upsilon_3(\cH_1):\\ 
& \begin{array}{|c||c|c|c|c|c|c|c|c|c|c|c|c|c|c|c|c|}
     x_{000} & 0 & 0 & 0 & 0 & 0 & 0 & 0 & 0 & 1 & 1 & 1 & 1 & 1 & 1 & 1 & 1 \\
     x_{001} & 0 & 0 & 1 & 1 & 0 & 0 & 1 & 1 & 0 & 0 & 1 & 1 & 0 & 0 & 1 & 1 \\
     x_{010} & 0 & 1 & 0 & 1 & 0 & 1 & 0 & 1 & 0 & 1 & 0 & 1 & 0 & 1 & 0 & 1 \\
     x_{011} & 0 & 1 & 1 & 0 & 1 & 0 & 0 & 1 & 0 & 0 & 0 & 1 & 1 & 1 & 1 & 0 \\
     x_{100} & 1 & 0 & 0 & 1 & 0 & 1 & 1 & 0 & 1 & 1 & 1 & 0 & 0 & 0 & 0 & 1 \\
     x_{101} & 1 & 0 & 1 & 0 & 1 & 0 & 1 & 0 & 1 & 0 & 1 & 0 & 1 & 0 & 1 & 0 \\
     x_{110} & 1 & 1 & 0 & 0 & 1 & 1 & 0 & 0 & 1 & 1 & 0 & 0 & 1 & 1 & 0 & 0 \\
     x_{111} & 1 & 1 & 1 & 1 & 1 & 1 & 1 & 1 & 0 & 0 & 0 & 0 & 0 & 0 & 0 & 0 \\
\end{array} \\ 
& \Upsilon_3(\cH_2):\\
& \begin{array}{|c||c|c|c|c|c|c|c|c|c|}
     x_{000} & 0 & 0 & 0 & 0 & 0 & 0 & 0 & 0 & 0 \\
     x_{001} & 0 & 0 & 1 & 0 & 0 & 1 & 0 & 0 & 0 \\
     x_{010} & 0 & 1 & 0 & 0 & 0 & 0 & 1 & 0 & 0 \\
     x_{011} & 0 & 1 & 1 & 0 & 0 & * & * & 0 & * \\
     x_{100} & 1 & 0 & 0 & 0 & 0 & 0 & 0 & 1 & 0 \\
     x_{101} & 1 & 0 & 1 & 0 & 0 & * & 0 & * & * \\
     x_{110} & 1 & 1 & 0 & 0 & 0 & 0 & * & * & * \\
     x_{111} & 1 & 1 & 1 & 0 & 1 & * & * & * & * 
\end{array}
\end{align*}

Every column in these tables represents a ternary polymorphism of $\cH_1,\cH_2$, respectively. Note that the first three polymorphisms in both cases are the projections, their values equal the first, second, and third component of the triple representing the label of variables of the indicator problem. Projections are polymorphisms of every structure, and therefore are always solutions of the indicator problem. 

The relation $\Upsilon_3(\cH_1)$ contains every polymorphism $f(x,y,z)$ satisfying the condition $f(\neg x,\neg y,\neg z)=\neg f(x,y,z)$, also known as the condition of \emph{self-duality}. The relation $\Upsilon_3(\cH_2)$ cannot be described by such a simple condition. The stars in the table for $\Upsilon_3(\cH_2)$ mean that no matter what values are in these positions, the tuple belongs to $\Upsilon_3(\cH_2)$.

The sets $I_{\cH_1}=I_{\cH_2}=\{000,011,100,111\}$ are equal in this case and so are $J_{\cH_1}=J_{\cH_2}=\{001,100\}$. The special tuples $\ba_{\cH_1},\bb_{\cH_1},\bc_{\cH_1},\bd_{\cH_1}$ and $\ba_{\cH_2},\bb_{\cH_2},\lb \bc_{\cH_2},\bd_{\cH_2}$ are also equal, as they are restrictions of the tuples corresponding to the same projections. Specifically,
\begin{align*}
& \ba_{\cH_1}=\ba_{\cH_2}=(0,0,1,1),\ \bb_{\cH_1}=\bb_{\cH_2}=(0,1,0,1),\\ 
& \bc_{\cH_1}=\bc_{\cH_2}=(0,1),\ \bd_{\cH_1}=\bd_{\cH_2}=(1,0). 
\end{align*}
As was observed, the presence of a Mal'tsev polymorphism is equivalent to the conditions 
\[
(\ba_{\cH_1},\bd_{\cH_1})\in\pr_{I_{\cH_1}\cup J_{\cH_1}}\Upsilon_3(\cH_1),\quad
(\ba_{\cH_2},\bd_{\cH_2})\in\pr_{I_{\cH_2}\cup J_{\cH_2}}\Upsilon_3(\cH_2),
\]
that is 
$
(0,0,1,1,1,0)\in \pr_{I_{\cH_1}\cup J_{\cH_1}}\Upsilon_3(\cH_1), \pr_{I_{\cH_2}\cup J_{\cH_2}}\Upsilon_3(\cH_2).
$
It is not difficult to see that such a tuple exists in $\pr_{I_{\cH_1}\cup J_{\cH_1}}\Upsilon_3(\cH_1)$, but not in $\pr_{I_{\cH_2}\cup J_{\cH_2}}\Upsilon_3(\cH_2)$.
\hbox{ }\hfill$\Box$
\end{example}

While $\Upsilon_3(\cH)$ characterizes the ternary polymorphisms of $\cH$, we now define a construction that characterizes the ternary polymorphisms of $\ang{\cH}_p$.

Note that every polymorphism of $\ang{\cH}_p$ is also a polymorphism of $\cH$. For any ternary polymorphism $f$ of $\cH$ that is not a polymorphism of $\ang{\cH}_p$, there exists a relation $\relo_f \in \ang{\cH}_p$ such that $f$ fails to preserve $\relo_f$. Let $\cH^{\dg p}$ be the expansion of $\cH$ by these relations 
\[
\cH^{\dg p}=\cH+\{\relo_f \mid f \in \Pol_3(\cH) - \Pol_3(\ang{\cH}_p)\}.
\]
The following lemma is straightforward.
\begin{lemma}\label{lem:indicator-definitions}
Let $\cH$ be a (multi-sorted) relational structure. Then
\begin{enumerate}
    \item $\Upsilon_3(\cH)$ is conjunctive-definable in $\cH$,
    \item $\Upsilon_3(\cH^{\dg p})$ is $p$-mpp-definable in $\cH$.
\end{enumerate}
\end{lemma}

If $\ang{\cH}_p$ does not have a Mal'tsev polymorphism, then $\cH^{\dg p}$ also does not have a Mal'tsev polymorphism. Therefore, 
\[
(\ba_\cH,\bc_\cH),(\bb_\cH,\bc_\cH),(\bb_\cH,\bd_\cH)\in\pr_{I_\cH\cup J_\cH}\Upsilon_3(\cH^{\dg p})
\] 
and $(\ba_\cH,\bd_\cH)\not\in\pr_{I_\cH\cup J_\cH}\Upsilon_3(\cH^{\dg p})$. 

A relation $\rel\sse\pr_{I_\cH\cup J_\cH}\Upsilon_3(\cH^{\dg p})$ such that $(\ba_\cH,\bc_\cH),(\bb_\cH,\bc_\cH),(\bb_\cH,\bd_\cH)\in\rel$ and $(\ba_\cH,\bd_\cH)\not\in\rel$ will be called a \emph{$p$-indicator rectangularity obstruction}. Clearly, if $\ang{\cH}_p$ has no Mal'tsev polymorphism, $\pr_{I_\cH\cup J_\cH}\Upsilon_3(\cH^{\dg p})$ itself is a $p$-indicator rectangularity obstruction. However, it may not be $p$-mpp-definable in $\cH$. A relational structure $\cH$ is said to \emph{admit a $p$-indicator rectangularity obstruction} if there exists a $p$-indicator rectangularity obstruction $\rel$ that is $p$-mpp-definable in~$\cH$. 

An important ingredient in our construction is the study of automorphisms of $\cH^{(3)}$ that fix certain structured tuples. Specifically, those of the form $(a,b,b)$ and $(b,b,a)$. We refer to such automorphisms as \emph{M-automorphisms}, with ``M'' standing for Mal'tsev. 

The motivation for working with $\cH^{(3)}$ comes from the need to apply our reduction tools in a higher-dimensional setting. While every automorphism of $\cH^{(3)}$ induces a polymorphism of $\cH$ via its coordinate projections, the converse does not hold: not every collection of polymorphisms arises as projections of a single automorphism of $\cH^{(3)}$. The notion of M-automorphisms captures precisely those symmetries of $\cH^{(3)}$ that preserve the key tuples required for our connection to Mal'tsev polymorphisms. This connection will become essential later.

\def\bxt{\widetilde{\bx}}
Let $\bx = (x_{i})_{i \in \cH^{(3)}}$ be the tuple of variables indexed by the elements of $H^{(3)}$. For a subset $K \subseteq \cH^{(3)}$, we use the notation $\bx_{\mid K}$ to denote the subtuple $(x_{i})_{i \in K}$, that is, the variables whose indices belong to $K$. The next lemma identifies conditions sufficient for $\cH$ to admit a $p$-indicator rectangularity obstruction. Define $E_\cH= H^{(3)}-(I_\cH\cup J_\cH)$ and $m=\ar(\Upsilon_3(\cH))$.


\begin{lemma}\label{lem:M-definitions}
Let $\rel(\bx_{|H^{(3)}-K})=\exists \bx_{|K}\,\, \Upsilon_3(\cH^{\dg p}) (\bx) \meet \Phi(\bx|_{K})$, where $K\sse E_\cH$, and $\Phi$ is a conjunction of unary predicates on variables from $K$ such that 
\[
(\ba_\cH,\bc_\cH),(\bb_\cH,\bc_\cH),(\bb_\cH,\bd_\cH)\in\pr_{I_\cH\cup J_\cH}\rel.
\]
Let also $y\in E_\cH - K$ and let 
\[
B_1 = \pr_y \Ext_\rel (\ba_\cH, \bc_\cH), \quad
B_2 = \pr_y \Ext_\rel (\bb_\cH, \bc_\cH), \quad
B_3 = \pr_y \Ext_\rel (\bc_\cH, \bd_\cH).
\]
If $B_1 \tm B_2 \tm B_3$ contains a triple that is a fixed point of every M-automorphism of $(\cH^{\dg p})^3$ then there is a $p$-mpp-definable relation $\relo\sse\pr_{H^{(3)}-(K\cup\{y\})}$ in $\ang{\cH+\{\rel\}}_p$ such that $(\ba_\cH,\bc_\cH),(\bb_\cH,\bc_\cH),(\bb_\cH,\bd_\cH)\in\pr_{I_\cH\cup J_\cH}\relo$, and 
\[
    |\pr_y(\Ext_{\relo}((\ba_\cH, \bc_\cH)))|, 
    |\pr_y(\Ext_{\relo}((\bb_\cH, \bc_\cH)))|, 
    |\pr_y(\Ext_{\relo}((\bb_\cH, \bd_\cH)))| \not \equiv 0 \pmod p.  \\
\]
\end{lemma}

\begin{proof}
Our goal is to construct a structure $(\cG, \vc xN, \vc yM, x)$ such that
\begin{align*}
\hom((\cG, \vc xN, \vc yM, x)&, (\cH, a_1, \dots, a_N, c_1 \zd c_M, B_1)) \not\equiv 0 \pmod p, \\
\hom((\cG, \vc xN, \vc yM, x)&, (\cH, b_1, \dots, b_N, c_1 \zd c_M, B_2)) \not\equiv 0 \pmod p, \\
\hom((\cG, \vc xN, \vc yM, x)&, (\cH, b_1, \dots, b_N, d_1 \zd d_M, B_3)) \not\equiv 0 \pmod p.
\end{align*}
To achieve this, we work in the power structure $\cH^3$. Consider the structure
\[
(\cH^3, (a_1,b_1,b_1), \dots, (a_N,b_N,b_N), (c_1,c_1,d_1), \dots, (c_M,c_M,d_M), B_1 \times B_2 \times B_3).
\]
By Lemma~\ref{lem:mobius-point-appendix}(3), there exists a structure $(\cG, \vc xN, \vc yM, x)$ such that
\begin{align}\label{equ:indicator-gadget}
\hom(&(\cG,\vc xN,\vc yM,x), \\
& (\cH^3,(a_1,b_1,b_1)\zd(a_N,b_N,b_N), (c_1,c_1,d_1)\zd(c_M,c_M,d_M), \nonumber\\
& B_1\tm B_2\tm B_3))\not\equiv 0\pmod p, \nonumber
\end{align}
By assumption, $B_1 \times B_2 \times B_3$ contains a triple $(\be_1, \be_2, \be_3)$ fixed by every $M$-automorphism. Thus, $\Stab((\be_1, \be_2, \be_3), B_1 \times B_2 \times B_3)$ forms an automorphism-stable set.

We now align the variables of $\rel$ with those of $\cG$: each variable of the form $(a_i,b_i,b_i)$ is identified with $x_i$, each $(c_i,c_i,d_i)$ with $y_i$, and $y$ with $x$. This alignment ensures that if the inputs to $\cG$ correspond to $(\ba_\cH, \bc_\cH)$, $(\bb_\cH, \bc_\cH)$, or $(\bb_\cH, \bd_\cH)$, then the corresponding extension of $y$ lies in $B_1$, $B_2$, or $B_3$, respectively. The structure $\cG$ guarantees that these extensions contribute non-zero counts modulo $p$.

Let $\{\vc t n\}$ be the auxiliary variables of $\cG$ not among $\vc xN, y_1\zd\lb y_M, x$. Then the relation
\[
\relo(\bx_{|H^{(3)} - K}) = \exists^{\equiv p} \vc t n\,\, \rel(\bx_{|H^{(3)} - K}) \meet \cG(\bx_{|I_\cH \cup J_\cH \cup \{ y \}}, \vc t n)
\]
is $p$-mpp-definable and satisfies the required properties.
\end{proof}

We complete this section with a simple observation.

\begin{lemma}\label{lem:M-auto-2-trivial}
Let $\cH=(\vc Hk;\vc\rel m)$ be a relational structure, $B=\{0,1\}\sse H_i$ its 2-element $p$-subalgebra for some $i\in[k]$, and $(f_1,f_2,f_3)$ an M-automorphism of $(\cH^{\dg p})^3$. Then either $(f_1,f_2,f_3)$ is the identity mapping on $B^3$ or $(f_1,f_2,f_3)(0,1,0)=(1,0,1), (f_1,f_2,f_3)(1,0,1)=(0,1,0)$. In particular, if $(f_1,f_2,f_3)$ has order $p>2$, then it is the identity mapping on $B^3$.
\end{lemma}

\begin{proof}
Since $B$ is a $p$-subalgebra, $(f_1,f_2,f_3)(B^3)=B^3$. Moreover every $(x,y,z)\in B^3-\{(0,1,0),(1,0,1)\}$ is a fixed point of $(f_1,f_2,f_3)$, leaving only the options listed in the lemma.
\end{proof}

\subsection{Conservative  Structures}\label{sec:maltsev-conservative}

In this section we use the results of Sections~\ref{sec:automorphic} and~\ref{sec:indicator} to essentially show that in the case of a $p$-conservative structure $\cH$ (single-sorted in the latter case) either there is a Mal'tsev polymorphism of $\ang{\cH}_p$ or $\cH$ admits a $p$-indicator rectangularity obstruction. The latter will later be used to prove the hardness of the corresponding modular counting problem.

Recall that a multi-sorted relational structure $\cH = (\{H_i\}_{i \in [k]}; \rel_1, \dots, \rel_m)$ is said to be \emph{$p$-conservative} if for every $i \in [k]$ and every subset $A \subseteq H_i$, the set $A$ is $p$-mpp-definable in $\cH$ (as a unary relation).

\begin{theorem}\label{the:conservative-malc}
Let $\cH$ be a $p$-conservative structure. Then either $\ang{\cH}_p$ has a Mal'tsev polymorphism, or $\cH$ admits a $p$-indicator rectangularity obstruction and therefore is not strongly $p$-rectangular.
\end{theorem}

\begin{proof}
We assume that $\ang{\cH}_p$ has no Mal'tsev polymorphism. By Lemma~\ref{lem:indicator-definitions}, it suffices to show that $\Upsilon_3(\cH^{\dg p})$ $p$-mpp-defines a $p$-indicator rectangularity obstruction.

Let $\vc zL$ be an ordering of the variables from the set $E_\cH$. We prove by induction that for any $i \in [L]$, there exists a relation $\rel'_i \subseteq \pr_{I_\cH \cup J_\cH \cup E_i} \Upsilon_3(\cH^{\dg p})$, where $E_i = \{z_i, \dots, z_L\}$, such that
\[
(\ba_\cH, \bc_\cH), (\bb_\cH, \bc_\cH), (\bb_\cH, \bd_\cH) \in \pr_{I_\cH \cup J_\cH} \rel'_i.
\]
Note that $\rel'_{L+1}$ is the desired $p$-indicator rectangularity obstruction. For the base case, let $\rel'_1 = \Upsilon_3(\cH^{\dg p})$, which satisfies the required condition, since $E_1 = \{\vc zL\}$.

Now suppose the statement holds for some $i$ and that $\rel'_i$ is the corresponding relation. Consider the variable $z_i \in E_i$, and let its domain be $H_j$. Define the sets $B_1, B_2, B_3 \subseteq H_j$ as the sets of values appearing at coordinate $z_i$ in the extensions of $(\ba_\cH, \bc_\cH)$, $(\bb_\cH, \bc_\cH)$, and $(\bb_\cH, \bd_\cH)$, respectively, within $\rel'_i$, i.e.,
\begin{align*}
B_1 &= \pr_{z_i} \Ext_{\rel'_i} (\ba_\cH, \bc_\cH), \\
B_2 &= \pr_{z_i} \Ext_{\rel'_i} (\bb_\cH, \bc_\cH), \\
B_3 &= \pr_{z_i} \Ext_{\rel'_i} (\bb_\cH, \bd_\cH).
\end{align*}
Let $D$ be a minimal set of representatives for $B_1$, $B_2$, and $B_3$; clearly, $|D| \le 3$. Let $\rel''=\rel'_i(\bx_{\mid I_\cH \cup J_\cH \cup E_i})\wedge D(z_i)$. 

\smallskip

\textsc{Case 1.} $|D| = 3$; that is, $B_1$, $B_2$, and $B_3$ are pairwise disjoint.

\smallskip

In this case, each of the tuples $(\ba_\cH, \bc_\cH)$, $(\bb_\cH, \bc_\cH)$, and $(\bb_\cH, \bd_\cH)$ has a distinct extension in $\pr_{z_i} \rel'_i$. In $\rel''$ every extension of $(\ba_\cH, \bc_\cH)$, $(\bb_\cH, \bc_\cH),(\bb_\cH, \bd_\cH)$ to a tuple from $\pr_{\bx_{\mid I_\cH \cup J_\cH}}\rel''$ has exactly one extension to a tuple from $\rel''$, and therefore
\[
\rel'_{i+1}(\bx_{\mid I_\cH \cup J_\cH\cup E_{i+1}})=\exists^{\equiv p} z_i\ \rel''
\]
satisfies the required conditions.

\smallskip

\textsc{Case 2.} $|D| = 2$.

\smallskip

Assume $D = \{0, 1\}$. Let $B'_1, B'_2, B'_3 \subseteq D$ be the sets of elements from $D$ that extend $(\ba_\cH, \bc_\cH)$, $(\bb_\cH, \bc_\cH)$, and $(\bb_\cH, \bd_\cH)$, respectively, in $\rel''$. If every $M$-automorphism of $\cH^{(3)}$ of order $p$ acts as the identity on $B'_1 \times B'_2 \times B'_3$, then the result follows from Lemma~\ref{lem:M-definitions}. Since $B'_1 \times B'_2 \times B'_3$ is a $p$-subalgebra of $\cH^{(3)}$ and therefore is  closed under automorphisms, this condition is satisfied, in particular, when $|B'_1| = |B'_2| = |B'_3| = 1$. Moreover, because $|D| = 2$, at most one of $B'_1, B'_2, B'_3$ can have size 2.

This leaves only three (non-symmetric) cases (up to swapping 0 and 1):
\begin{itemize}
    \item[(a)] $B'_1 = \{0,1\},\ B'_2 = \{0\},\ B'_3 = \{1\}$,
    \item[(b)] $B'_1 = \{0\},\ B'_2 = \{0,1\},\ B'_3 = \{1\}$,
    \item[(c)] $B'_1 = \{0\},\ B'_2 = \{1\},\ B'_3 = \{0,1\}$.
\end{itemize}

In case (a), since $(0,0,1)$ is a fixed point of any $M$-automorphism, it follows that $(1,0,1)$ is also fixed. Hence, every $M$-automorphism must be the identity on $B'_1 \times B'_2 \times B'_3$. The same argument applies in case (c). In case (b), we have $B'_1 \times B'_2 \times B'_3 = \{(0,0,1), (0,1,1)\}$, and every $M$-automorphism must fix this set pointwise as well.

\smallskip

\textsc{Case 3.} $|D| = 1$.

\smallskip

This means that there exists an element $c \in H_j$ such that $c \in B_1 \cap B_2 \cap B_3$. As in Case~1 every extension of $(\ba_\cH, \bc_\cH)$, $(\bb_\cH, \bc_\cH),(\bb_\cH, \bd_\cH)$ to a tuple from $\pr_{{ I_\cH \cup J_\cH}}\rel''$ has exactly one extension to a tuple from $\rel''$. We again can set $\rel'_{i+1}(\bx_{\mid I_\cH \cup J_\cH\cup E_{i+1}})=\exists^{\equiv p} z_i\ \rel''$.
\end{proof}

\subsection{3-Element Structures}\label{sec:3-elem-maltsev}

Next, we assume $\cH$ to be a 3-element (single-sorted) structure $(H,\vc\rel\ell)$, where $H=\{0,1,2\}$. The main result of this section is the following theorem.

\begin{theorem}\label{the:3-elem-maltsev}
Let $\cH$ be a 3-element structure. Then one of the following three options holds
\begin{itemize}
\item[(1)]
$\ang{\cH}_p$ has a Mal'tsev polymorphism and therefore $\cH$ is strongly $p$-rectangular;
\item[(2)]
$\cH$ admits a $p$-indicator rectangularity obstruction;
\item[(3)]
$p=2$ and $\cH$ has a $p$-automorphic polynomial $f$, there is $a\in H$ such that $f(a,x)$ is a permutation of order~2, and $H-\{a\}$ is a $p$-subalgebra of $\cH$.
\end{itemize}
\end{theorem}

\begin{proof}
We follow the same approach as in the proof of Theorem~\ref{the:conservative-malc} and assume that $\ang{\cH}_p$ has no Mal'tsev polymorphism. Under this assumption we either show how a $p$-indicator rectangularity obstruction can be $p$-mpp-defined using $\Upsilon_3(\cH^{\dg p})$ or construct a 2-automorphic polynomial of $\cH$. 
We reuse the notation for $E_\cH=\{\vc zL\}$ and argue by induction on $i\in[L+1]$ that there exists a relation 
$\rel'_i\sse\pr_{I_\cH\cup J_\cH\cup E_i}\Upsilon_3(\cH^{\dg p})$, where $E_i=\{z_i\zd z_L\}$, such that $\rel'_i$ is $p$-mpp-definable in $\cH$ and $(\ba_\cH,\bc_\cH),(\bb_\cH,\bc_\cH),(\bb_\cH,\bd_\cH)\in\pr_{I_\cH\cup J_\cH}\rel'_i$. Again,  $\rel'_{L+1}$ is the required $p$-indicator rectangularity obstruction and $\rel'_1=\Upsilon_3(\cH^{\dg p})$ provides  the base case of induction. 

Suppose $\rel'_i$ is constructed. Consider $z_i\in E_\cH$ and let $B_1,B_2,B_3$ be the extensions of $(\ba_\cH,\bc_\cH)$, $(\bb_\cH,\bc_\cH)$, $(\bb_\cH,\bd_\cH)$, respectively, to the coordinate $z_i$. Since any triple of the form $(a,b,b)$ or $(b,b,a)$ is a fixed point of any M-automorphism of $\cH^{(3)}$, if $B_1\tm B_2\tm B_3$ contains a tuple of this form, $z_i$ can be quantified away by Lemma~\ref{lem:M-definitions}. The only three cases when it does not happen are: 
\begin{enumerate}
    \item $|B_1|=|B_2|=|B_3|=1$, and all $B_1,B_2,B_3$ are distinct, 
    \item $|B_1|=|B_3|=1$, $B_1=B_3$, $|B_2|=2$ and $B_1\cap B_2=B_3\cap B_2=\eps$,
    \item $|B_1|=|B_3|=2$, $B_1=B_3$, $|B_2|=1$ and $B_1\cap B_2=B_3\cap B_2=\eps$.
\end{enumerate}

\smallskip

{\sc Case 1.}
$|B_1|=|B_2|=|B_3|=1$, and all $B_1,B_2,B_3$ are distinct.

\smallskip

In this case we can set 
\[
\rel'_{i+1}(\bx_{\mid I_\cH \cup J_\cH \cup E_{i+1}})=\exists^{\equiv p} z_i\ \rel'_i.
\]
\smallskip

{\sc Case 2.}
$|B_1|=|B_3|=1$, $B_1=B_3$, $|B_2|=2$ and $B_1\cap B_2=B_3\cap B_2=\eps$.

\smallskip

Without loss of generality, $B_1=B_3=\{2\}$ and $B_2=\{0,1\}$. Note that as $B_2$ is the set of all extensions of a certain tuple and constant relations are in $\cH$, it must be a $p$-subalgebra of $\cH^{\dg p}$ preserved by all the polymorphisms. If $B_1\tm B_2\tm B_3=\{202,212\}$ contains a fixed point of all M-automorphisms of $\cH^3$ then $z_i$ can be quantified away by Lemma~\ref{lem:M-definitions}. Otherwise there is an M-automorphism $(g_1,g_2,g_3)$ of $\cH^3$ swapping 202 and 212 (and thus $p$ must be 2). We show that in this case the polymorphisms $g_1,g_2,g_3$ of $\cH$ generate an automorphic polynomial of order 2.  Set $f_1(x,y)=g_1(x,y,x),f_2(x,y)=g_2(x,y,x),f_3(x,y)=g_3(x,y,x)$. The operations $f_1,f_2,f_3$ preserve $\{0,1\}$ and, as we assume that $(g_1,g_2,g_3)$ swap 202 and 212, we obtain that $f_1(2,0)=f_1(2,1)=f_3(2,0)=f_3(2,1)=2$, $f_2(2,0)=1$, and $f_2(2,1)=0$. Thus, they have unknown values only on $(1,0),(0,1),(0,2),(1,2)$. We use the fact that $(g_1,g_2,g_3)$ is an M-automorphism of $\cH^3$. For instance, it means that there are only 2 options on how it acts on $(0,1,0)$ and $(1,0,1)$: they are either fixed points or being swapped.  Tables~1 and~2 below list all the possibilities for the values of $g_1,g_2,g_3$ on the triples $(0,1,0),(1,0,1),(0,2,0),(1,2,1)$ and a term in $f_1,f_2,f_3$ that generates the 2-automorphic polynomial $f$ given by its operation table 
\[
\begin{array}{l|ccc}
f&0&1&2\\
\hline
0&0&1&2\\
1&0&1&2\\
2&1&0&2
\end{array}
\]
The terms are found using the Universal Algebra Calculator \cite{Freese:ucalc}. Arrows show the images of tuples under $(g_1,g_2,g_3)$.

\smallskip

{\sc Subcase 2.1.} $010\leftrightarrow101$, $202\leftrightarrow212$\\[3mm]
We exploit symmetries: in all the terms below, if a term works for a M-automorphism mapping $202\to x_1y_1z_1, 121\to x_2y_2z_2$, then swapping $f_1$ and $f_3$ will work for a M-automorphism mapping $202\to z_1y_1x_1, 121\to z_2y_2x_2$. 
\bgroup\setlength{\belowdisplayskip}{-1ex}
\begin{gather*}
\begin{array}{||l|l||}
\hline
020\to 121\hspace{3mm} 121\to 020  & f_1(y,f_1(x,y))\\
020\to 012\hspace{3mm} 121\to 121  & f_2(f_3(x,y),y)\\
020\to 102\hspace{3mm} 121\to 121  & f_2(f_1(x,y),y)\\
020\to 020\hspace{3mm} 121\to 121  & f_1(f_2(x,y),f_1(x,y))\\
020\to 012\hspace{3mm} 121\to 021  & f_1(y,f_2(y,x))\\
020\to 102\hspace{3mm} 121\to 021  & f_1(y,f_1(x,y))\\
020\to 201\hspace{3mm} 121\to 021  & f_2(f_3(x,y),y)\\
020\to 210\hspace{3mm} 121\to 021  & f_1(y,f_2(y,x))\\
020\to 020\hspace{3mm} 121\to 021  & f_2(f_1(x,y),y)\\
020\to 102\hspace{3mm} 121\to 012  & f_1(y,f_1(x,y))\\
020\to 201\hspace{3mm} 121\to 012  & f_1(y,f_2(f_1(x,y),y))\\
020\to 210\hspace{3mm} 121\to 012  & f_1(y,f_1(f_3(y,x),f_2(y,x)))\\
020\to 020\hspace{3mm} 121\to 012  & f_2(f_1(x,y),y)\\
020\to 012\hspace{3mm} 121\to 102  & f_1(y,f_2(y,x))\\
020\to 201\hspace{3mm} 121\to 102  & f_1(y,f_2(f_1(x,y),y))\\
020\to 210\hspace{3mm} 121\to 102  & f_1(y,f_2(y,x))\\
020\to 020\hspace{3mm} 121\to 102  & f_1(y,f_1(f_2(x,y),f_1(x,y)))\\
\hline
\end{array}
\end{gather*}
\captionof{table}{\hbox{\vbox to 7mm{}}Subcase 2.1. $010\leftrightarrow101$, $202\leftrightarrow212$}
\egroup

\medskip

{\sc Subcase 2.2.}
$010\leftrightarrow010$, $101\leftrightarrow101$, $202\leftrightarrow212$\\[3mm]
We again use symmetries in the same way as in Case~2.1. The stars in the first 3 rows indicate that the term applies to any combination of these values. 
Also, by line 1 of Table~2 one of $g_2(020),g_2(121)$ should not be 2. We assume it is $g_2(020)$, the other case is symmetric by swapping 0 and 1.
\bgroup\setlength{\belowdisplayskip}{-1ex}
\begin{gather*}
\begin{array}{||l|l||}
\hline
020\to *2*\hspace{3mm} 121\to *2*  & f_2\\
020\to 1**\hspace{3mm} 121\to 0** & f_1\\
020\to **1\hspace{3mm} 121\to **0 & f_3\\
020\to 012\hspace{3mm} 121\to 121  & f_2(f_3(x,y),y)\\
020\to 102\hspace{3mm} 121\to 121  & f_2(f_1(x,y),y)\\
020\to 012\hspace{3mm} 121\to 021  & f_2(f_3(x,y),y)\\
020\to 201\hspace{3mm} 121\to 021  & f_1(f_2(x,y),y)\\
020\to 210\hspace{3mm} 121\to 021  & f_2(y,f_1(y,x))\\
020\to 201\hspace{3mm} 121\to 012  & f_1(f_1(f_2(x,y),y),y)\\
020\to 210\hspace{3mm} 121\to 012  & f_2(y,f_1(y,x))\\
020\to 012\hspace{3mm} 121\to 102  & f_2(y,f_3(x,y))\\
020\to 201\hspace{3mm} 121\to 102  & f_1(f_2(x,y),y)\\
020\to 210\hspace{3mm} 121\to 102  & f_2(f_1(y,x),f_3(y,x))\\
\hline
\end{array}
\end{gather*}
\captionof{table}{\hbox{\vbox to 7mm{}}Subcase 2.2. $010\leftrightarrow010$, $101\leftrightarrow101$, $202\leftrightarrow212$}
\egroup

\medskip

{\sc Case 3.}
$|B_1|=|B_3|=2$, $B_1=B_3$, $|B_2|=1$ and $B_1\cap B_2=B_3\cap B_2=\eps$.

\smallskip

Without loss of generality, $B_1=B_3=\{0,1\}$ and $B_2=\{2\}$. Note that $B_1,B_3$ must be $p$-subalgebras of $\cH^{\dg p}$ preserved by all the polymorphisms. If $B_1\tm B_2\tm B_3=\{020,021,120,121\}$ contains a fixed point of all M-automorphisms of $\cH^3$ then $z_i$ can be quantified away by Lemma~\ref{lem:M-definitions}. Since $|B_1\tm B_2\tm B_3|=4$, this is the case when $p>3$. If $p=3$, then observe that every extension of $(\ba_\cH,\bc_\cH),(\bb_\cH,\bc_\cH),(\bb_\cH,\bd_\cH)$ in $\pr_{I_\cH\cup J_\cH\cup E_{i+1}}\rel'_i$ has 1 or 2 extensions to a tuple from $\rel'_i$. Therefore, we can set 
\[
\rel'_{i+1}(\bx_{|I_\cH\cup J_\cH\cup E_{i+1}})=\exists^{\equiv3}z_i\ \rel'_i.
\]

So we assume $p=2$. In this case we show that the polymorphisms $g_1,g_2,g_3$ of $\cH$ generate an automorphic  polynomial of order 2. Set again $f_1(x,y)=g_1(x,y,x),\lb f_2(x,y)=g_2(x,y,x),f_3(x,y)=g_3(x,y,x)$. The operations $f_1,f_2,f_3$ preserve $\{0,1\}$, therefore there are only 2 options for the action of $g_1,g_2,g_3$ on $\{0,1\}^3$: $010\leftrightarrow101$ and the identity mapping. Along with the condition that $g_1,g_2,g_3$ preserves $B_1\tm B_2\tm B_3$ and that it is an M-automorphism of order 2, the only undefined values are on $202,212,012,102,201,210$. We will need to look into the latter 4 in only one case. Tables~3--6 below list all the possibilities for the values of $f_1,f_2,f_3$ (and in one case of $g_1,g_2,g_3$ as well) and a term that generates the automorphic polynomial 
\[
\begin{array}{l|ccc}
f&0&1&2\\
\hline
0&0&1&2\\
1&0&1&2\\
2&1&0&2
\end{array}
\]

\smallskip

{\sc Subcase 3.1.}$010\leftrightarrow101$, $020\leftrightarrow121$\\[3mm]
We exploit symmetries: in all the terms below, if a term works for a M-automorphism mapping $202\to x_1y_1z_1, 121\to x_2y_2z_2$, then swapping $f_1$ and $f_3$ will work for a M-automorphism mapping $202\to z_1y_1x_1, 121\to z_2y_2x_2$. 
\bgroup\setlength{\belowdisplayskip}{-1ex}
\begin{gather*}
\begin{array}{||l|l||}
\hline
202\to 212\hspace{3mm} 212\to 202  & f_1(y,f_1(x,y))\\
202\to 012\hspace{3mm} 212\to 212  & f_1(y,f_2(y,x))\\
202\to 102\hspace{3mm} 212\to 212  & f_2(f_1(x,y),f_3(y,x))\\
202\to 202\hspace{3mm} 212\to 212  & f_1(y,f_1(x,y))\\
202\to 012\hspace{3mm} 212\to 021  & f_1(y,f_2(y,f_2(y,x)))\\
202\to 102\hspace{3mm} 212\to 021  & f_2(f_1(x,y),y)\\
202\to 201\hspace{3mm} 212\to 021  & f_1(y,f_2(y,x))\\
202\to 210\hspace{3mm} 212\to 021  & f_2(f_1(x,y),y)\\
202\to 202\hspace{3mm} 212\to 021  & f_1(y,f_2(y,x))\\
202\to 102\hspace{3mm} 212\to 012  & f_1(y,f_1(y,f_1(x,y)))\\
202\to 201\hspace{3mm} 212\to 012  & f_1(y,f_1(y,f_1(x,y)))\\%
202\to 210\hspace{3mm} 212\to 012  & f_2(f_1(x,y),y)\\
202\to 202\hspace{3mm} 212\to 012  & f_2(f_1(x,y),f_3(y,x))\\
202\to 012\hspace{3mm} 212\to 102  & f_2(f_2(y,x),y)\\
202\to 201\hspace{3mm} 212\to 102  & f_1(y,f_2(y,x))\\
202\to 210\hspace{3mm} 212\to 102  & f_2(f_2(y,x),y)\\
202\to 202\hspace{3mm} 212\to 102  & f_1(y,f_2(y,x))\\
\hline
\end{array}
\end{gather*}
\captionof{table}{\hbox{\vbox to 7mm{}}Subcase 3.1.$010\leftrightarrow101$, $020\leftrightarrow121$}
\egroup

\smallskip

{\sc Subcase 3.2.} $010\leftrightarrow010$, $101\leftrightarrow101$, $020\leftrightarrow121$\\[3mm]
We again use exploit symmetries in the same way as before.
Also, by lines 2,3 of the following table one of $g_1(020),g_1(121)$ and one of $g_3(020),g_3(121)$ should not be 2. This means that none of 202, 212 is mapped to 202 or 212. We assume it is $g_2(020)$, the other case is symmetric by swapping 0 and 1.
\bgroup\setlength{\belowdisplayskip}{-1ex}
\begin{gather*}
\begin{array}{||l|l||}
\hline
202\to *1*\hspace{3mm} 212\to *0*  & f_2\\
202\to 2**\hspace{3mm} 212\to 2** & f_1\\
202\to **2\hspace{3mm} 212\to **2 & f_3\\
202\to 201\hspace{3mm} 212\to 012  & f_1(y,f_1(x,f_1(y,f_3(x,y))))\\
202\to 210\hspace{3mm} 212\to 012  & f_3(y,f_2(y,x))\\
202\to 201\hspace{3mm} 212\to 102  & f_1(y,f_2(x,y))\\
\hline
\end{array}
\end{gather*}
\captionof{table}{\hbox{\vbox to 7mm{}}Subcase 3.2. $010\leftrightarrow010$, $101\leftrightarrow101$, $020\leftrightarrow121$}
\egroup

\smallskip

{\sc Subcase 3.3.} $010\leftrightarrow101$, $020\to 021,\ \ 121\to 120$\\[3mm]
Under the same assumptions

\bgroup\setlength{\belowdisplayskip}{-1ex}
\begin{gather*}
\begin{array}{||l|l||}
\hline
202\to *1*\hspace{3mm} 212\to *0*  & f_2\\
202\to 202\hspace{3mm} 212\to212 & f_1(f_3(y,x),f_1(x,y))\\
202\to 202\hspace{3mm} 212\to201 & f_3(y,f_2(y,x))\\
202\to 202\hspace{3mm} 212\to210 & f_1(y,f_1(f_3(y,x),f_1(x,y)))\\
202\to 202\hspace{3mm} 212\to012 & f_2(f_1(x,y),f_3(y,x))\\
202\to 202\hspace{3mm} 212\to102 & f_1(f_3(y,x),f_2(y,x))\\
202\to 212\hspace{3mm} 212\to012 & f_2(f_1(x,y),y)\\
202\to 212\hspace{3mm} 212\to210 & f_1(f_2(x,y),f_3(x,y))\\
202\to 201\hspace{3mm} 212\to202 & f_1(f_2(x,y),f_3(x,y))\\
202\to 201\hspace{3mm} 212\to212 & f_1(y,f_1(f_3(y,x),f_1(x,y)))\\
202\to 201\hspace{3mm} 212\to210 & f_1(y,f_1(f_3(y,x),f_1(x,y)))\\
202\to 201\hspace{3mm} 212\to012 & f_1(y,f_1(f_3(y,x),f_1(x,y)))\\
202\to 201\hspace{3mm} 212\to102 & f_2(f_3(x,y),y)\\
202\to 210\hspace{3mm} 212\to212 & f_3(y,f_2(y,x))\\
202\to 210\hspace{3mm} 212\to102 & f_2(f_1(x,y),y)\\
202\to 012\hspace{3mm} 212\to212 & f_1(f_3(y,x),f_2(y,x))\\
202\to 012\hspace{3mm} 212\to210 & f_2(f_3(x,y),y)\\
202\to 102\hspace{3mm} 212\to202 & f_2(f_1(x,y),y)\\
202\to 102\hspace{3mm} 212\to212 & f_2(f_1(x,y),f_3(y,x))\\
202\to 102\hspace{3mm} 212\to201 & f_2(f_1(x,y),y)\\
202\to 102\hspace{3mm} 212\to210 & f_1(y,f_1(f_3(y,x),f_1(x,y)))\\
202\to 102\hspace{3mm} 212\to012 & f_2(f_1(x,y),y)\\
\hline
\end{array}
\end{gather*}
\captionof{table}{\hbox{\vbox to 7mm{}}Subcase 3.3. $010\leftrightarrow101$, $020\to 021,\ \ 121\to 120$}
\egroup
\smallskip

{\sc Subcase 3.4.} $010\leftrightarrow010$, $101\leftrightarrow101$, $020\to 021,\ \ 121\to 120$\\[3mm]
This is the only case when we also need the operations $g_1,g_2,g_3$ themselves.
\bgroup\setlength{\belowdisplayskip}{-1ex}
\begin{gather*}
\begin{array}{||l|l||}
\hline
202\to **2\hspace{3mm} 212\to **2 & f_3\\
202\to *1*\hspace{3mm} 212\to *1* & f_2\\
202\to 202\hspace{3mm} 212\to201 & f_3(y,f_2(y,x))\\
202\to 202\hspace{3mm} 212\to210\ \ 201\ \ 012\to 102 & g_1(y,g_3(y,x,y),x)\\
202\to 202\hspace{3mm} 212\to210\ \ 012\ \ 201\to 102 & g_3(g_3(y,x,y),x,y)\\
202\to 202\hspace{3mm} 212\to210\ \ 102\ \ 201\to 012 & g_1(g_3(y,x,y),x,y)\\
202\to 202\hspace{3mm} 212\to210\ \ 102\ \ 201\ \ 012 & g_3(g_3(x,y,x),x,g_3(x,y,y))\\
202\to 201\hspace{3mm} 212\to212\ \ 210\ \ 012\to102 & g_1(y,g_3(y,x,y),x)\\
202\to 201\hspace{3mm} 212\to212\ \ 012\ \ 210\to102 & g_1(g_3(y,x,y),x,y)\\
202\to 201\hspace{3mm} 212\to212\ \ 102\ \ 210\to012 & g_3(g_3(y,x,y),x,y)\\
202\to 201\hspace{3mm} 212\to212\ \ 102\ \ 210\ \ 012 & g_3(g_3(y,x,y),x,y)\\
202\to 201\hspace{3mm} 212\to210\ \  012\to 102 & g_1(y,g_3(x,y,x),x)\\
202\to 201\hspace{3mm} 212\to210\ \  012\ \  102 & g_3(g_3(x,y,x),x,y)\\
202\to 201\hspace{3mm} 212\to012\ \ 210\to 102 & g_2(g_3(x,y,x),y,g_1(x,y,x))\\
202\to 201\hspace{3mm} 212\to012\ \ 210\ \ 102 & g_3(g_3(y,x,y),x,y)\\
202\to 201\hspace{3mm} 212\to102\ \ 210\to 012 & g_1(g_3(y,x,y),x,g_1(y,x,y))\\
202\to 201\hspace{3mm} 212\to102\ \ 210\ \ 012 & g_1(g_3(y,x,y),x,g_1(y,x,y))\\
202\to 210\hspace{3mm} 212\to212 & f_3(y,f_2(y,x))\\
202\to 210\hspace{3mm} 212\to012 & f_3(y,f_2(y,x))\\
202\to 012\hspace{3mm} 212\to210 \ \ 201\to 102 & g_1(g_3(y,x,y),x,g_1(y,x,y))\\
202\to 012\hspace{3mm} 212\to210 \ \ 201\ \ 102 & g_1(g_3(y,x,y),x,g_1(y,x,y))\\
202\to 102\hspace{3mm} 212\to201 & f_3(y,f_2(y,x))\\
202\to 102\hspace{3mm} 212\to210\ \ 201\to 012 & g_2(g_3(x,y,x),y,g_1(x,y,x))\\
202\to 102\hspace{3mm} 212\to210\ \ 201\ \ 012 & g_3(g_3(y,x,y),x,y)\\
\hline
\end{array}
\end{gather*}
\captionof{table}{\hbox{\vbox to 7mm{}}Subcase 3.4. $010\leftrightarrow010$, $101\leftrightarrow101$, $020\to 021,\ \ 121\to 120$}
\egroup
\end{proof}

\section{Complexity Classifications: Conservative Structures}\label{sec:conservative-classification}

In this section we give an (incomplete) complexity classification of $\#_p\CSP(\cH)$, where $\cH$ is a $p$-conservative structure. The missing case is $p>2$ and $\ang{\cH}_p$ has a Mal'tsev polymorphism. The gap is due the lack of understanding of the complexity of partition functions over nontrivial finite fields. It is a highly nontrivial problem that requires a separate study, see \cite{Kazeminia25:thesis} for some partial results.

We start with proving hardness for structures whose domains are 2-element.

\begin{proposition}\label{pro:2-element-obstruction-appendix}
Let $\cH$ be a multi-sorted relational structure, all of whose domains contain at most 2 elements. If a $p$-indicator rectangularity obstruction is $p$-mpp-definable in $\cH$, then $\#_p\CSP(\cH)$ is $\#_pP$-hard.
\end{proposition}

\begin{proof}
Suppose $\rel \subseteq \pr_U \Upsilon_3(\cH)$ is a $p$-mpp-definable $p$-indicator rectangularity obstruction, where $U \subseteq I_\cH \cup J_\cH$, $|U| \ge 2$, and both $I' = U \cap I_\cH$ and $J' = U \cap J_\cH$ are nonempty. Define the projections
\[
\ba' = \pr_{I'} \ba_\cH,\quad \bb' = \pr_{I'} \bb_\cH,\quad \bc' = \pr_{J'} \bc_\cH,\quad \bd' = \pr_{J'} \bd_\cH,
\]
and assume $(\ba', \bc'), (\bb', \bc'), (\bb', \bd') \in \rel$ while $(\ba', \bd') \notin \rel$.

If $|U| = 2$, then since each domain has at most 2 elements, the relation $\rel$ corresponds to a binary relation over two 2-element domains. In this case, $\#_p\CSP(\rel)$ is equivalent to evaluating the partition function of a matrix
\[
M = \begin{pmatrix} 0 & 1 \\ 1 & 1 \end{pmatrix},
\]
which is known to be $\#_p\mathrm{P}$-complete by~\cite{faben2008complexity}. Thus, the result holds in this case.

We now consider the case $|U| \ge 3$ and aim to reduce the arity of $\rel$ while preserving the rectangularity obstruction. That is, we want to find a subset $U' \subset U$ and a $p$-mpp-definable relation $\rel' \subseteq \pr_{U'} \rel$ that still encodes the obstruction, i.e.,
\[
(\ba'', \bc''),\, (\bb'', \bc''),\, (\bb'', \bd'') \in \rel',\quad (\ba'', \bd'') \notin \rel',
\]
where $\ba'' = \pr_{I''} \ba_\cH$, $\bb'' = \pr_{I''} \bb_\cH$, $\bc'' = \pr_{J''} \bc_\cH$, $\bd'' = \pr_{J''} \bd_\cH$, and $I'' = U' \cap I_\cH$, $J'' = U' \cap J_\cH$ are nonempty.

Since $|U| \ge 3$, at least one of $|I'| \ge 2$ or $|J'| \ge 2$ must hold. Without loss of generality, assume $|J'| \ge 2$ and select a variable $v \in J'$. Write $\bc' = (\bc'', 0)$ and $\bd' = (\bd'', 1)$, arranging the coordinates so that $v$ corresponds to the last position.

Let $B_1, B_2, B_3$ be the sets of possible extensions of $(\ba', \bc''), (\bb', \bc''), (\bb', \bd'')$ in $\rel$, respectively, over the coordinate $v$. Since the domains have at most 2 elements, each $B_i \subseteq \{0,1\}$ has size at most 2.

If $p \ge 3$, then regardless of the structure of the $B_i$, we may define $U' = U - \{v\}$ and let
\[
\rel'(\bx_{|U'}) = \exists^{\equiv p} \bx_{|v}\, \rel(\bx_{|U}) = \pr_{U'} \rel,
\]
which is still $p$-mpp-definable and preserves the obstruction (as the indicator condition is unaffected modulo $p$). This gives the desired result.

Now suppose $p = 2$ and we analyze based on the intersections of the $B_i$. 

{\sc Case 1:} $B_1 \cap B_2 \cap B_3 \neq \emptyset$. Let $c \in B_1 \cap B_2 \cap B_3$. Then we define
\[
\rel'(\bx_{|U'}) = \exists^{\equiv 2} \bx_{|v}\, (\rel(\bx_{|U}) \wedge C_c(\bx_{|v}),
\]
which restricts $v$ to a constant and ensures the projection of $\rel'$ satisfies the required rectangularity obstruction over $U' = U - \{v\}$.

{\sc Case 2:} $B_1 \cap B_2 \ne \emptyset$ and $B_2 \ne B_3$. Without loss of generality we can assume that $B_1 = \{ 1,0\}, B_2 = \{ 0\}$ and $B_3=\{ 1\}$, suppose $1 \in B_1 - B_2$. Let $U'' = I' \cup \{v\}$ and define
\[
\rel'(\bx_{|U''}) = \exists^{\equiv 2} \bx_{|J' - \{v\}} \, (\rel(\bx_{|U}) \wedge \bigwedge_{j \in J' - \{ v \}} C_{\bc''_{|j}}(\bx_{|j}).
\]
Then $\rel'$ contains $(\ba', 0), (\ba', 1), (\bb', 0)$ but not $(\bb', 1)$, that is, it is a rectangularity obstruction. Although this is a “swapped” obstruction (i.e., $\ba'$ and $\bb'$ trade roles), the structure still suffices to encode hardness, after relabeling.

{\sc Case 3:} $B_1 = B_2 \ne B_3 $. Without loss of generality $B_1 = B_2 = \{ 1\}, B_3 = \{0\}$. Then the value of $v$ is fixed for all tuples in $\rel$ that match the patterns, and we may again let
\[
\rel'(\bx_{|U - \{ v \}}) = \exists^{\equiv 2} \bx_{|v} \, \rel(\bx_{|U}) = \pr_{U - \{v\}} \rel,
\]
which satisfies the same type of rectangularity obstruction.
We now argue that the case analysis covers all possible configurations of $B_1, B_2, B_3 \subseteq \{0,1\}$ with $B_i \ne \emptyset$. There are only seven relevant patterns (up to symmetry):

\begin{itemize}
    \item $B_1 = B_2 = B_3$: handled in Case~1.
    \item $B_1 = B_2 \ne B_3$, with $B_1 \cap B_3 \ne \emptyset$: falls under Case~1.
    \item $B_1 = B_2 \ne B_3$, with $B_1 \cap B_3 = \emptyset$: handled in Case~3.
    \item $B_1 = \{0,1\}, B_2 = \{0\}, B_3 = \{1\}$: handled in Case~2.
    \item All $B_i$ different but pairwise intersecting: reduces to Case~1.
    \item Two equal, third intersects both: reduces to Case~1.
    \item Any configuration where $B_1 \cap B_2 \cap B_3 \ne \emptyset$: handled in Case~1.
\end{itemize}

Since $\{0,1\}$ has only two elements, no three nonempty subsets can be pairwise disjoint. Therefore, the three cases exhaust all possibilities.

By iterating this reduction process, we eventually arrive at a relation of arity 2, and as argued above, such a relation yields a $\#_p\mathrm{P}$-complete problem. Therefore, $\#_p\CSP(\cH)$ is $\#_p\mathrm{P}$-hard.
\end{proof}

\begin{proposition}\label{pro:2-element-obstruction}
Let $\cH$ be a multi-sorted relational structure, all of whose domains contain at most 2 elements. 
The problem $\#_p\CSP(\cH)$ is solvable in polynomial time if and only if $\ang{\cH}_p$ has a Mal'tsev polymorphism. Otherwise it is $\#_pP$-hard.
\end{proposition}

\begin{proof}
Note that a relational structure with constant relations, each of whose domains is 2-element, is conservative. Therefore, by Theorem~\ref{the:conservative-malc} if $\ang{\cH}_p$ does not have a Mal'tsev polymorphism then the $p$-indicator rectangularity obstruction is $p$-mpp-definable in $\cH$. By Proposition~\ref{pro:2-element-obstruction-appendix} $\#_p\CSP(\cH)$ is $\#_pP$-hard.

Suppose that $\ang{\cH}_p$ has a Mal'tsev polymorphism $m$. From the description of clones on a 2-element set \cite{Post41} we can conclude that for every domain $H_i$ of $\cH$ there is a polymorphism $m_i$ of $\ang{\cH}_p$ such that $m_i$ acts as $x\oplus y\oplus z$ (or the minority operation) on $H_i$, where $\oplus$ is addition modulo~2. If we prove that there exists a polymorphism $m'$ of $\ang{\cH}_p$ that acts as a minority operation on every domain of $\cH$ then every relation of $\ang{\cH}_p$ is representable by a system of linear equations modulo~2, and therefore $\#\CSP(\cH)$ is solvable in polynomial time, and so is $\#_p\CSP(\cH)$.

For each $H_i=\{a_i,b_i\}$, there are 4 options of how the operation $m$ can act. Indeed, the values 
\[
m(a_i,a_i,b_i),m(a_i,b_i,b_i),m(b_i,a_i,a_i),m(b_i,b_i,a_i)
\]
are fixed, as $f$ is a Mal'tsev operation that only leaves $m(a_i,b_i,a_i),m(b_i,a_i,b_i)\in\{a_i,b_i\}$. Depending on these values we call $m$ on $H_i$ type~0 if $m(a_i,b_i,a_i)=b_i$ and $m(b_i,a_i,b_i)=a_i$ (in this case $m$ on $H_i$ equals the minority operation $x\oplus y\oplus z$, exactly what we want to obtain); $m$ on $H_i$ is of type~1 if $m(a_i,b_i,a_i)=a_i$ and $m(b_i,a_i,b_i)=b_i$; of type~2 if $m(a_i,b_i,a_i)=m(b_i,a_i,b_i)=a_i$; and of type~3 if $m(a_i,b_i,a_i)=m(b_i,a_i,b_i)=b_i$.

Now, consider the term
\[
f(x,y,z)=m(m(x,y,z),x,m(x,z,y)).
\]
As is easily seen, $f$ acts as a minority operation whenever $m$ on $H_i$ is of type~1,2, or~3, and it is the first projection $f(x,y,z)=x$ if $m$ on $H_i$ is of type~0. This means that
\[
g(x,y)=f(x,x,y)
\]
acts as the first projection if $m$ on $H_i$ is of type~0, and as the second projection otherwise. Finally, consider
\[
h(x,y,z)=g(m(x,y,z),f(x,y,z)).
\]
This operation acts as $m(x,y,z)$ if $m$ on $H_i$ is of type~0, and it acts as $f(x,y,z)$ otherwise. In both cases $h$ is the minority operation.
\end{proof}

\begin{theorem}\label{the:conservative-hard}
Let $\cH$ be a conservative structure and $p$ a prime. If $\ang{\cH}_p$ does not have a Mal'tsev polymorphism then $\#_p\CSP(\cH)$ is $\#_pP$-hard. If $\ang{\cH}_2$ has a Mal'tsev polymorphism then $\#_2\CSP(\cH)$ is solvable in polynomial time.
\end{theorem}

\begin{proof}
The second statement of the theorem follows from Theorem~15 of \cite{Kazeminia25:modular}. For the first statement, by Theorem~\ref{the:conservative-malc} a $p$-indicator rectangularity obstruction $\rel$ is $p$-mpp-definable in $\cH$. Recall that 
$
(\ba_\cH,\bc_\cH),(\bb_\cH,\bc_\cH),(\bb_\cH,\bd_\cH)\in\rel,
$
while $(\ba_\cH,\bd_\cH)\not\in\rel$. As $\cH$ is conservative, for every $v\in I_\cH$ ($v\in J_\cH$), $H_{i_v}=\{\ba[v],\bb[v]\}$ ($H_{i_v}=\{\bc[v],\bd[v]\}$) is a $p$-subalgebra of $\cH$. Set
\[
\rel'(\bx_{|I_\cH \cup J_\cH})=\rel(\bx_{|I_\cH \cup J_\cH})\meet\bigwedge_{v\in I_\cH\cup J_\cH} H_{i_v}(\bx_{v}).
\]
The relation $\rel'$ still contains $(\ba_\cH,\bc_\cH),(\bb_\cH,\bc_\cH),(\bb_\cH,\bd_\cH)$ and does not contain  $(\ba_\cH,\bd_\cH)$. By Proposition~\ref{pro:2-element-obstruction} $\#_p\CSP(\rel')$ is $\#_pP$-hard.
\end{proof}

\section{Complexity Classification: 3-Element Structures}\label{sec:3-element-classification}

Next, we consider the case when $\cH$ is a 3-element structure.

\begin{theorem}\label{the:dichotomy-3-element}
Let $\cH$ be a 3-element structure and $p$ a prime. If one of the following conditions holds:
\begin{enumerate}
\item[(a)] 
$\cH$ has a $p$-automorphic polynomial $f$ and $\ang{\cH^f}_p$ has a Mal'tsev polymorphism, or 
\item[(b)] 
$\cH$ does not have a $p$-automorphic polynomial, $\ang{\cH}_p$ has a Mal'tsev polymorphism, and $p=2$,
\end{enumerate}
then $\#_p\CSP(\cH)$ is solvable in polynomial time. Also, if
\begin{enumerate}
\item[(1)] 
$\cH$ has a $p$-automorphic polynomial $f$ and $\ang{\cH^f}_p$ does not have a Mal'tsev polymorphism, or
\item[(2)] 
$\cH$ does not have a $p$-automorphic polynomial and $\ang{\cH}_p$ does not have a Mal'tsev polymorphism, 
\end{enumerate}
then $\#_p\CSP(\cH)$ is $\#_pP$-complete.
\end{theorem}

In order to prove Theorem~\ref{the:dichotomy-3-element} we need to prove some auxiliary statements first.



\begin{lemma}\label{lem:automorphism-poly-automorphism-appendix}
Let $\cH$ be a 3-element structure with $H=\{a,b,c\}$ and $\{a,b\}$ its $p$-subalgebra. Then if $\cH^2$ has an automorphism that swaps $(c,a)$ and $(c,b)$, then it also has a 2-automorphic polynomial.
\end{lemma}

\begin{proof}
Let $g(x,y)=(g_1(x,y),g_2(x,y))$ be an automorphism specified in the lemma. Then the operations $g_1,g_2$ are defined on $(a,a),(b,b),(c,c),(c,a),(c,b)$ and they map $(a,b),(b,a)$ to $\{a,b\}$, as it is a $p$-subalgebra. Then $g$ either swaps $(a,b),(b,a)$ or leave them in place, and it either swaps $(a,c),(b,c)$ or leaves them in place. If $(a,b),(b,a)$ are fixed points, then $g_1(x,y)$ is an automorphic polynomial. Otherwise an automorphic polynomial can be obtained as is shown in the table below.
\[
\begin{array}{||l|l||}
\hline
ab\leftrightarrow ba\hspace{3mm} ac\to ac\hspace*{3mm} bc\to bc &	g_1(g_2(x,y),g_1(x,y))\\
ab\leftrightarrow ba\hspace{3mm} ac\leftrightarrow bc &  g_1(x,g_1(y,x))\\
\hline
\end{array}
\]
\end{proof}


\begin{lemma}\label{lem:3-element-eval-appendix}
Let $\cH$ be a 3-element structure without automorphic polynomials and $\rel$ a binary non-rectangular relation $p$-mpp-definable in $\cH$. Then $\#_p\CSP(\cH)$ is $\#_pP$-hard.
\end{lemma}

\begin{proof}
Suppose that $H=\{a,b,c\}=\{a',b',c'\}$ (i.e.\ $a',b',c'$ is a permutation of $a,b,c$), and $(a,b'),(b,b'),(b,a')\in\rel$ and $(a,a')\not\in\rel$. Then $\#_p\CSP(\rel)$ is equivalent to the problem $\EVAL_p(M)$, where $M$ is the matrix
\[
M=\left(\begin{array}{ccc} 0&1&*\\ 1&1&*\\ *&*&* \end{array}\right),
\]
where $*$ stands for an undetermined value. We consider $M$ as the adjacency matrix of a bipartite graph $G$. If $p>3$ then $G$ has no $p$-automorphisms, and therefore $\EVAL_p(M)$ is $\#_pP$-hard by the results of \cite{DBLP:conf/stoc/BulatovK22}. If $p=3$ and $G$ has no 3-automorphism then we again get the result by \cite{DBLP:conf/stoc/BulatovK22}. If $G$ has a 3-automorphism, then the only option for $M$ is 
\[
M=\left(\begin{array}{ccc} 0&1&1\\ 1&1&0\\ 1&0&1 \end{array}\right).
\]
In this case the formula $\exists^{\equiv3} y (\rel(x,y)\wedge C_{b'}(y)$ shows that $B=\{a,b\}$ is a $p$-subalgebra of $\cH$. Therefore the relation $\rel(x,y)\wedge B(x)$ is equivalent to $\EVAL(M')$ with 
\[
M'=\left(\begin{array}{ccc} 0&1\\ 1&1\\ 1&1 \end{array}\right),
\]
which is also $\#_3P$-hard by \cite{DBLP:conf/stoc/BulatovK22}.

Let $p=2$. If the last two columns and the last two rows of $M$ are different, $G$ also has no 2-automorphisms, and therefore  and $\EVAL_2(M)$ is $\#_2P$-hard. If the last two columns of $M$ are equal, then by pinning $\rel(a,x)$ we obtain that $\{b,c\}$ is a $p$-subalgebra of $\cH$. By Lemma~\ref{lem:automorphism-poly-automorphism-appendix}, as $\cH$ has no automorphic polynomials, $\cH^2$ also has no 2-automorphism that swaps $(a,b)$ and $(b,c)$. By the same argument there is no 2-automorphism that swaps $(b',a')$ and $(c',a')$. Then Theorem~1.4 of \cite{DBLP:conf/stoc/BulatovK22} implies that $\EVAL_2(M)$ is $\#_2P$-hard. 
\end{proof}


\begin{proof}[Proof of Theorem~\ref{the:dichotomy-3-element}]
Let $H=\{0,1,2\}$. If $\cH$ has an automorphic polynomial, the result follows from Corollary~\ref{cor:3-element-auto-poly} and Proposition~\ref{pro:2-element-obstruction}. If $\ang{\cH}_p$ has a Mal'tsev polymorphism and $p=2$, the result follows from Theorem~15 of \cite{Kazeminia25:modular}. Therefore, we assume that $\ang{\cH}_p$ has no Mal'tsev polymorphism and $\cH$ does not have an automorphic polynomial. By Theorem~\ref{the:3-elem-maltsev} the standard rectangularity obstruction is $p$-mpp-definable in $\cH$. We prove that it is possible to $p$-mpp-define a binary non-rectangular relation showing that $\#_p\CSP(\cH)$ is $\#_pP$-complete. We consider four cases depending on the number of $p$-subalgebras $\cH$ has.

\medskip

\noindent
{\sc Case 1.}
$\cH$ is conservative.

\medskip

In this case the result follows from Theorem~\ref{the:conservative-hard}.

\medskip

In the remaining cases, let $U\sse I_\cH\cup J_\cH$ such that $|U|\ge3$, $I'=U\cap I_\cH\ne\eps$, $J'=U\cap J_\cH\ne\eps$, and $\rel \sse\pr_U\Upsilon_3(\cH^{\dg p})$ such that $\rel$ is $p$-mpp-definable in $\cH$ and $(\ba',\bc'),(\bb',\bc'),(\bb',\bd')\in\rel$, $(\ba',\bd')\not\in\rel$, where $\ba'=\pr_{I'}\ba_\cH, \bb'=\pr_{I'}\bb_\cH, \bc'=\pr_{J'}\bc_\cH, \bd'=\pr_{J'}\bd_\cH$. We prove that in every case there is a set $U'\subset U$ and a relation $\rel'\sse\pr_{U'}\rel$, such that $I''=U'\cap I_\cH\ne\eps$, $J''=U'\cap J_\cH\ne\eps$, $\rel'$ is $p$-mpp-definable in $\cH$ and $(\ba'',\bc''),(\bb'',\bc''),(\bb'',\bd'')\in\rel'$, $(\ba'',\bd'')\not\in\rel'$, where $\ba''=\pr_{I''}\ba_\cH, \bb''=\pr_{I''}\bb_\cH, \bc''=\pr_{J''}\bc_\cH, \bd''=\pr_{J''}\bd_\cH$. 

\medskip

\noindent
{\sc Case 2.}
$\cH$ has two 2-element $p$-subalgebras.

\medskip

As $|U|\ge3$, either $|I'|\ge2$ or $|J'|\ge2$. Without loss of generality suppose $|J'|\ge2$ and pick $v\in J'$. Let $\bc'=(\bc'',0)$ and $\bd'=(\bd'',1)$, where $v$ corresponds to the last coordinate (rename the elements of $H$ if necessary). Also, let $B_1,B_2,B_3$ be the sets of extensions of $(\ba',\bc''),(\bb',\bc''),(\bb',\bd'')$ in $\rel$. Suppose first $B_1\ne B_2$, say, $a\in B_1-B_2$. Then set $U'=I'\cup\{v\}$ and
\[
\rel'(\bx_{|U'})=\exists^{\equiv p}\bx_{J'-\{v\}}\ \rel(\bx_{|U})\meet \bigwedge_{j \in J' - \{ v \}} C_{\bc''_{|j}}(\bx_{|j}).      
\]
The relation $\rel'$ contains $(\ba',0),(\ba',a),(\bb',0)$, but does not contain $(\bb',a)$. If $B_1=B_2$ and $B_1\cap B_3\ne\eps$, say, $a\in B_1\cap B_3$ then we set $U'=U-\{v\}$ and
\[
\rel'(\bx_{|U'})=\exists^{\equiv p}\bx_{|v}\ \rel(\bx_{|U})\meet C_a(\bx_{|v}).
\]

In the remaining case one option is $|B_1|=|B_2|=|B_3|=1$, in which case we set $U'=U-\{v\}$ and
\[
\rel'(\bx_{|U'})=\exists^{\equiv p}\bx_{|v}\ \rel(\bx_{|U}).
\]
The other option is that there is a 2-element $p$-subalgebra $D$ of $\cH$ such that $|B_1\cap D|=|B_2\cap D|=|B_3\cap D|=1$, in which case we set $U'=U-\{v\}$ and
\[
\rel'(\bx_{|U'})=\exists^{\equiv p}\bx_{v}\ \rel(\bx_{|U}) \meet D(\bx_{|v}).
\]

\medskip

\noindent
{\sc Case 3.}
$\cH$ has one 2-element $p$-subalgebra.

\medskip

Let us assume that $D=\{a,b\}$ is the only 2-element $p$-subalgebra. Since $|U|\ge3$, either $|I'|\ge2$ or $|J'|\ge2$; assume the latter. Pick a coordinate $v\in J'$ as follows: if $J'$ contains a coordinate $v$ such that $\bd_\cH[v]\not\in D,\bc_\cH[v]\in D$, then pick a coordinate like this. Otherwise choose an arbitrary $v\in J'$. As in Case~2 let $\bc'=(\bc'',0)$ and $\bd'=(\bd'',1)$, where $v$ corresponds to the last coordinate of $\bc',\bd'$. Let $B_1,B_2,B_3$ be the sets of extensions of $(\ba',\bc''),(\bb',\bc''),(\bb',\bd'')$ in~$\rel$. 

If $B_1\cap B_2\cap B_3\ne\eps$, say, $b\in B_1\cap B_2\cap B_3$ then we choose $U'=U-\{v\}$ and  
\[
\rel'(\bx_{|U'})=\exists^{\equiv p}\bx_{|v}\ \rel(\bx_{|U})\meet C_b(\bx_{|v}).
\]
Therefore we may assume that $B_3=\{1\}$ or $B_3=C$ and $0\not\in C$.

Suppose now that $B_1\ne B_2$, say, $c\in B_1-B_2$, $c\in\{1,2\}$. Then set $U'=I'\cup\{v\}$ and
\[
\rel'(\bx_{|U'})=\exists^{\equiv p}(\bx_{|J'-\{v\}})\ \rel(\bx_{|U})\meet \bigwedge_{j \in J' - \{ v \}} C_{\bc''_{|j}}(\bx_{|j}).
\]
The relation $\rel'$ contains $(\ba',0),(\ba',c),(\bb',0)$, but does not contain $(\bb',c)$. 

If $B_3=\{1\}$ and $B_1=B_2=\{0\}$ then set $U'=U-\{v\}$ and $\rel'=\exists^{\equiv p}v\ \rel$. If $B_1=B_2=\{0,2\}=C$, then by Lemma~\ref{lem:automorphism-poly-automorphism-appendix} there is no $p$-automorphism of $\cH^2$ that swaps $(1,0)$ and $(1,2)$. By Lemma~3.7 of \cite{DBLP:conf/stoc/BulatovK22} there exists a structure $(\cG,x)$ such that $\hom((\cG,x),(\cH,\{1\})),\hom((\cG,x),(\cH,C))\not\equiv0\pmod p$. Let $V$ be the set of vertices of $\cG$. Then set $U'=U-\{v\}$ and 
\[
\rel'(\bx_{|K'})=\exists^{\equiv p} V \ (\rel\meet\cG\meet(x=v)).
\]

Finally, suppose that $B_3=C$ and $0\not\in C$. Note that in this case $B_1=B_2=\{0\}$ (as otherwise it has to be $H$). Then as before by Lemma~3.7 of \cite{DBLP:conf/stoc/BulatovK22} there exists a structure $(\cG,x)$ such that $\hom((\cG,x),(\cH,\{0\})),\hom((\cG,x),(\cH,C))\not\equiv0\pmod p$ and we complete the proof as in the previous case.
\medskip

\noindent
{\sc Case 4.}
$\cH$ has no 2-element $p$-subalgebras.

\medskip

Since $|U|\ge3$, either $|I'|\ge2$ or $|J'|\ge2$; assume the latter and pick an arbitrary $v\in J'$. Let $\bc'=(\bc'',0)$ and $\bd'=(\bd'',1)$, where $v$ corresponds to the last coordinate of $\bc',\bd'$, and let $B_1,B_2,B_3$ defined as in Case~3. Then for each $i\in[3]$, either $|B_i|=1$ or $|B_i|=3$. Thus, unless $p=3$ we can set $U'=U-\{v\}$ and $\rel'(U')=\exists^{\equiv p} v\ \rel(U)$. The same construction works when $p=2$ and $|B_1|=|B_2|=|B_3|=1$. As in Case~3, if $B_1\cap B_2\cap B_3\ne\eps$, say, $c\in B_1\cap B_2\cap B_3$ then we choose $U'=U-\{v\}$ and  
\[
\rel'(\bx_{|U'})=\exists^{\equiv p}\bx_{|v}\ \rel(\bx_{|U})\meet C_b(\bx_{|v}).
\]
Also, if $B_1\ne B_2$, say, $c\in B_1-B_2$, $c\in\{1,2\}$. Then set $U'=I'\cup\{v\}$ and
\[
\rel'(\bx_{|U'})=\exists^{\equiv p}(\bx_{|J'-\{v\}})\ \rel(\bx_{|U})\meet \bigwedge_{j \in J' - \{ v \}} C_{\bc''_{|j}}(\bx_{|j}).
\]
As is easily seen, these cover all the possibilities.
\end{proof}

\bibliographystyle{plainurl}
\bibliography{refrences.bib}

\end{document}